\documentclass[lettersize,journal]{IEEEtran}

\usepackage{amsmath,amsfonts,amssymb,amsthm,amscd,array,latexsym}
\usepackage[caption=false,font=normalsize,labelfont=sf,textfont=sf]{subfig}
\usepackage{textcomp}
\usepackage{stfloats}
\usepackage{url}
\usepackage{verbatim}
\usepackage{nccmath}
\usepackage{graphicx}
\usepackage{cite}
\usepackage[hidelinks]{hyperref}      
\usepackage[resetlabels]{multibib}
\usepackage{url}            
\usepackage{booktabs}       
\usepackage{nicefrac}      
\usepackage{placeins}
\usepackage{multirow}
\usepackage{bm}
\usepackage{siunitx}
\usepackage{enumitem}
\usepackage[ruled]{algorithm}
\usepackage[]{algorithmicx}
\usepackage[noend]{algpseudocode}
\usepackage[dvipsnames]{xcolor}
\usepackage[percent]{overpic}
\usepackage{standalone}
\usepackage{xcolor}
\usepackage{soul}
\usepackage{enumerate}
\usepackage{mathrsfs}
\usepackage{overpic}
\usepackage{mathtools}
\usepackage{graphicx,subfig,wrapfig}
\usepackage{times}
\usepackage{psfrag,epsfig}
\usepackage{comment}
\usepackage{tabularx}
\usepackage{setspace}
\usepackage{color}

\DeclareMathAlphabet{\mathpzc}{OT1}{pzc}{m}{it}

\newtheorem{theorem}{Theorem}

\newtheorem{lemma}{Lemma}
\newtheorem{definition}{Definition}

\newtheorem{remark}{Remark}
\newtheorem{prp}{Proposition}

\newtheorem*{prf*}{Proof}

\def\Ce{\mathbb{C}}
\def\Ic{\mathcal{I}}
\def\Lc{\mathcal{L}}

\def\Es{\mathscr{E}}

\DeclareMathOperator{\im}{im}

\DeclareMathOperator{\diag}{diag}
\DeclareMathOperator{\snr}{\mathrm{SNR}}
\DeclareMathOperator{\nmse}{\mathrm{nMSE}}
\DeclareMathOperator{\wa}{\mathrm{WA}}
\DeclareMathOperator{\db}{\si{\dB}}

\newcommand{\ts}{\textsuperscript}

\DeclareFontFamily{U}{BOONDOX-calo}{\skewchar\font=45 }
\DeclareFontShape{U}{BOONDOX-calo}{m}{n}{
  <-> s*[1.05] BOONDOX-r-calo}{}
\DeclareFontShape{U}{BOONDOX-calo}{b}{n}{
  <-> s*[1.05] BOONDOX-b-calo}{}
\DeclareMathAlphabet{\mathcalboondox}{U}{BOONDOX-calo}{m}{n}
\SetMathAlphabet{\mathcalboondox}{bold}{U}{BOONDOX-calo}{b}{n}
\DeclareMathAlphabet{\mathbcalboondox}{U}{BOONDOX-calo}{b}{n}

\def\Qb{{\boldsymbol{\mathcalboondox{Q}}}}

\begin{document}

\title{Shuffled Multi-Channel Sparse Signal Recovery}

\author{Taulant Koka, Manolis C. Tsakiris, Michael Muma and Benjamín Béjar Haro\thanks{Taulant Koka and Michael Muma are with the Robust Data Science Group at Technische Universit{\"a}t Darmstadt, Germany (e-mail: taulant.koka@tu-darmstadt.de; michael.muma@tu-darmstadt.de).}
\thanks{Manolis C. Tsakiris is with the Academy of Mathematics and Systems Science at the
Chinese Academy of Sciences, China (e-mail: manolis@amss.ac.cn).}
\thanks{Benjamín Béjar Haro is with the Swiss Data Science Center at the Paul Scherrer Institute, Switzerland (e-mail: benjamin.bejar@psi.ch).}\thanks{The work of the first and second author has been funded in part by the ERC Starting Grant ScReeningData under grant number 101042407.}}

\markboth{SUBMITTED TO IEEE TRANSACTIONS ON SIGNAL PROCESSING}
{Shell \MakeLowercase{\textit{et al.}}: Bare Demo of IEEEtran.cls for IEEE Journals}

\maketitle

\begin{abstract}
Mismatches between samples and their respective channel or target commonly arise in several real-world applications. For instance, whole-brain calcium imaging of freely moving organisms, multiple-target tracking or multi-person contactless vital sign monitoring may be severely affected by mismatched sample-channel assignments. To systematically address this fundamental problem, we pose it as a signal reconstruction problem where we have lost correspondences between the samples and their respective channels. Assuming that we have a sensing matrix for the underlying signals, we show that the problem is equivalent to a structured unlabeled sensing problem, and establish sufficient conditions for unique recovery. To the best of our knowledge, a sampling result for the reconstruction of shuffled multi-channel signals has not been considered in the literature and existing methods for unlabeled sensing cannot be directly applied. We extend our results to the case where the signals admit a sparse representation in an overcomplete dictionary (i.e., the sensing matrix is not precisely known), and derive sufficient conditions for the reconstruction of shuffled sparse signals. We propose a robust reconstruction method that combines sparse signal recovery with robust linear regression for the two-channel case. The performance and robustness of the proposed approach is illustrated in an application related to whole-brain calcium imaging. The proposed methodology can be generalized to sparse signal representations other than the ones considered in this work to be applied in a variety of real-world problems with imprecise measurement or channel assignment.
\end{abstract}


\IEEEpeerreviewmaketitle

\section{Introduction}
\label{sec:introduction}
\IEEEPARstart{T}{he} problem of reconstructing a signal without precise knowledge about the sample locations has recently received considerable attention from the research community, e.g.,  \cite{unnikrishnan_unlabeled_2018,dokmanic_permutations_2019,tsakiris_algebraic-geometric_2020,slawski_pseudo-likelihood_2021,tsakiris_homomorphic_2019,ashwin_pananjady_linear_2018,hsu_linear_2017,abid_stochastic_2018,peng_linear_2020,elhami_unlabeled_2017,peng_homomorphic_2021,xu_uncoupled_2019,abbasi_r-local_2021,yao_unlabeled_2021,wang_estimation_2020,slawski_linear_2019,abid_linear_2017,TSAKIRIS2023210,onaran2022shuffled,abbasi_r-local_2022}. In a discrete setting, such a problem is usually referred to as {\em unlabeled sensing} or {\em shuffled linear regression}. 

In shuffled linear regression, the sensing matrix is assumed to be known and the goal is to recover the right ordering (permutation matrix) together with the regression coefficients that generated the observations. Here, we consider a related problem where the shuffling of samples does not happen within the signal but rather across multiple channels of a multivariate signal (see Figs.~\ref{fig:general_setup} and \ref{fig:shuffled}), thus adding more structure to the problem due to the swapping of samples across channels only. Although the problem of imprecise sample-channel assignments is well-known within the literature and is addressed, for example, by probabilistic data association filters for multiple-target tracking \cite{Shalom2009probabilistic}, we take a fundamentally different approach, where we assume that the signals admit a parametric representation in a {\em generic} subspace, and look at the problem from a sampling perspective. To the best of our knowledge, such an approach has not been previously considered in the literature. The considered setup, broadly applies to situations where some quantity of interest is measured from multiple moving targets that may be difficult to tell apart due to partial occlusions, crossings or measurement noise (Fig.~\ref{fig:general_setup}~(a)). 
For instance, entomological radar for observing insect flight and migration \cite{long_entomological_2020}, fluorescence microscopy of moving cell cultures \cite{schenk_high-speed_2016}, contactless vital sign monitoring of groups of people or animals \cite{brijs_remote_2019}, \cite{liu_multi-task_2020},\cite{schroth2023emergency}, multiple-target tracking \cite{junkun_target_2021,teklehaymanot_adaptive_2017} or whole-brain calcium imaging of freely moving organisms such as {\em zebrafish} or {\em C. elegans}~\cite{cong_rapid_2017} constitute some paradigmatic examples (Fig.~\ref{fig:general_setup}~(b))

As it turns out, the problem of reconstructing shuffled multi-channel signals coincides with an unlabeled sensing formulation, where both the sensing matrix and the permutation matrix are highly structured. This structure, however, prevents existing unlabeled sensing results from being directly applicable, necessitating the development of new theory within the scope of this paper. We expand on the developed theory by also addressing the case, where the exact sensing matrix is not precisely known {\em a priori}, making our setup more general in the sense that the sensing matrix has to be inferred from the observations. Of course, without any prior knowledge about the signals of interest, the problem is hopeless. In our case, we rely on the assumption that each of the signals of interest (each channel) admits a sparse representation in an overcomplete dictionary with possibly infinitely many elements. This allows us to first retrieve the support of the signals and estimate the sensing matrix that can subsequently be used to solve the sample assignment and signal reconstruction problems.

\paragraph*{Main Contributions}(i) We formalize the shuffled multi-channel signal reconstruction problem as a structured unlabeled sensing problem, and provide sufficient conditions for unique recovery. (ii) We extend the recovery results to the case where the sensing matrix is not precisely known a priori by assuming that the signals admit a sparse representation on a dictionary over the continuum. (iii) We propose a two-step method for shuffled sparse signal reconstruction in the two-channel case that combines robust regression with sparse signal recovery to determine the correct sample-channel assignment and regression coefficients.
(iv) We demonstrate and benchmark the performance of the proposed methodology in numerical experiments for sparse signal reconstruction, and showcase practical applicability in experiments related to whole-brain calcium imaging in neuroscience. A python implementation of the proposed method will be made available on the authors github page.
\paragraph*{Organization}Section~\ref{sec:related_work} reviews the related work before the structured unlabeled sensing setup is introduced in Section~\ref{sec:sus}, where we derive conditions for unique recovery. Section~\ref{sec:sss} then builds upon these results to formalize a sampling theorem for the reconstruction of sparse shuffled signals, and a robust method for reconstruction of such signals in the two-channel case is established in Section~\ref{ssec:rob_est}. Finally, the numerical evaluation of the proposed methodology is detailed in Section~\ref{sec:numerical_experiments}. We leave Section~\ref{sec:conclusion} for conclusions and further directions.
\paragraph*{Notation} We use capital letters to denote frequency domain representations, calligraphic letters to denote sets, bold-faced calligraphic letters to denote tensors, and bold-faced letters to denote vectors and matrices, i.e., an $N$-dimensional vector is denoted as $\bm x = ( x_0,\: x_2,...,\:x_{N-1})^\top$ and an $N\times N$ matrix is denoted as $\bm X = \begin{pmatrix}\bm x_0&\bm x_1&\dots&\bm x_{N-1}\end{pmatrix}$. The $(i,j)$ entry of $\bm X$ is denoted as $[\bm X]_{ij}$. For a set of indexed elements $\{\Omega_1,\dots,\Omega_M\}$ we use $\{\Omega_m\}_{m=1}^M$ as shorthand. For vector spaces $\mathcal{V}$, $\mathcal{W}$, we write $\dim \mathcal{V}$ and $\dim \mathcal{W}$ for their respective dimensions, $\mathcal{V}\cap\mathcal{W}$ for their intersection, and $\mathcal{V}\cup\mathcal{W}$ for their union. We denote by $\ker(\varrho)$ the kernel and by $\im(\varrho)$ the image of some linear map $\varrho:\mathcal{V}\rightarrow\mathcal{W}$.
\begin{figure*}[t]
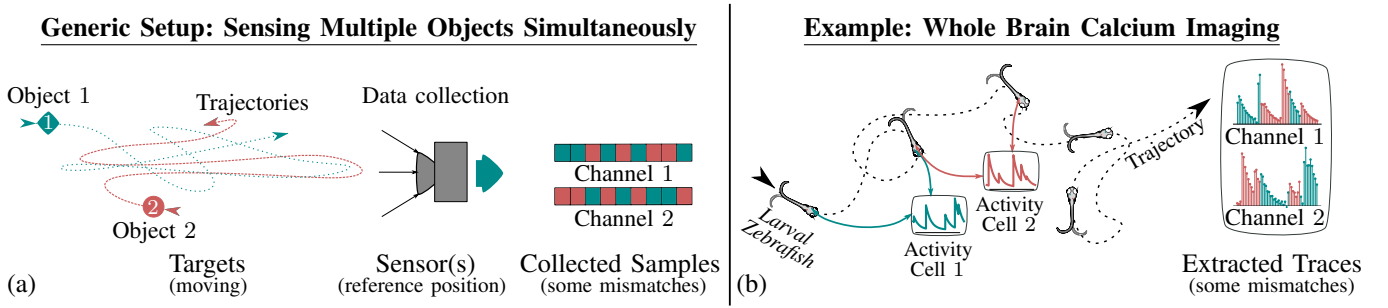

    \centering
     \begin{overpic}[width=0.99\textwidth,tics=2]{Figs/general_example_setup_no_text.pdf}
        \put (0,1) {(a)}
        \put (54,1) {(b)}
        \put (2.5,20) {\underline{\textbf{Generic Setup: Sensing Multiple Objects Simultaneously}}}
        \put (59,20) {\underline{\textbf{Example: Whole Brain Calcium Imaging}}}
        \put (2.65,13.05) {{\color{white}{\footnotesize $1$}}}
        \put (0,15) {\small Object $1$}
        \put (10.4,6.8) {{\color{white}{\footnotesize $2$}}}
        \put (7.8,5) {\small Object $2$}
        \put (12,2.5) {Targets}
        \put (12,1) {\footnotesize(moving)}
        \put (27.3,2.5) {Sensor(s)}
        \put (24.5,1) {\footnotesize(reference position)}
        \put (42,9.2) {\small Channel $1$}
        \put (42,5.9) {\small Channel $2$}
        \put (26.3,15) {\small Data collection}
        \put (38,2.5) {Collected Samples}
        \put (39.5,1) {\footnotesize (some mismatches)}
        \put (14,14.5) {\small Trajectories}
        \put (90.2,12) {\small Channel $1$}
        \put (90.2,6) {\small Channel $2$}
        \put (82.7,10) {\rotatebox{30}{\footnotesize Trajectory}}
        \put (55.6,6.5) {\rotatebox{-30}{\footnotesize {\em Larval}}}
        \put (54.2,5.6) {\rotatebox{-30}{\footnotesize {\em Zebrafish}}}
        \put (66.4,3.7) {\footnotesize Activity}
        \put (66.9,2.2) {\footnotesize Cell $1$}
        \put (71.7,7) {\footnotesize Activity}
        \put (72.2,5.5) {\footnotesize Cell $2$}
        \put (87,2.5) {Extracted Traces}
        \put (87.5,1) {\footnotesize (some mismatches)}
    \end{overpic}
    \caption{\textbf{Settings without precise sample-channel assignments:} (a) Illustration of the generic setting in which acquired samples of some multivariate quantity of interest (Object $1$ and $2$) may be shuffled (Channel $1$ and $2$) due to crossings, occlusions, or measurement noise. (b) Real-world application example in whole-brain calcium imaging. Traces of neuronal activity may be shuffled, due to wrongly annotated neurons.}
    \label{fig:general_setup}
\end{figure*}

\section{Related Work}
\label{sec:related_work}
We start by providing a short literature review of related work on the topics of unlabeled sensing and sampling and reconstruction of continuous-time sparse signals.

\paragraph*{Unlabeled Sensing}There has been a considerable amount of recent research regarding linear regression without correspondences, which can be formulated as a linear system of equations
$ \boldsymbol{y} = \boldsymbol{\Pi}\boldsymbol{A}\boldsymbol{\beta},\:$ with the permutation matrix $\bm \Pi$ and coefficients $\bm\beta$ as unknowns.
This was subsequently extended from permutations to any invertible and diagonalizable matrix \cite{dokmanic_permutations_2019}, as well as to arbitrary linear maps \cite{tsakiris_homomorphic_2019}. Further, it was shown that the maximum likelihood estimate (MLE) of $\boldsymbol{\beta}$ for a response vector that is corrupted by additive random noise tends towards $\boldsymbol{\beta}$ for an increasing signal-to-noise ratio (SNR) \cite{unnikrishnan_unlabeled_2018} and the authors in \cite{hsu_linear_2017} established lower bounds on the SNR below which the estimation error is above a threshold for any estimator. On the other hand, it was shown in  \cite{ashwin_pananjady_linear_2018} that the  ground-truth permutation matrix $\boldsymbol{\Pi}$ coincides with high probability with the associated MLE if the SNR is fixed and exceeds a threshold. However, if the SNR is too low, the MLE differs from $\boldsymbol{\Pi}$ with high probability and it could be shown for a fixed SNR that the MLE is asymptotically inconsistent \cite{abid_linear_2017}. For a sparsely shuffled response vector, it was shown that under mild hypotheses $\boldsymbol{\beta}$ coincides with the optimal solution of an $\ell_1$-regression problem which can be solved via convex optimization \cite{slawski_linear_2019}, and an improvement of the methodology by relying on hypothesis testing, expectation maximization (EM) and reweighted least squares, was made in \cite{slawski_pseudo-likelihood_2021}. Other EM approaches were proposed in \cite{abid_stochastic_2018} and improved upon in \cite{tsakiris_algebraic-geometric_2020}, where the authors developed an algebraic geometric theory for this problem by expressing it as a polynomial system, which contains $\boldsymbol{\beta}$ in its root locus. Approaches based on branch-and-bound, RANSAC and concave minimization have also been proposed in \cite{tsakiris_homomorphic_2019}, \cite{peng_linear_2020}. In \cite{peng_homomorphic_2021}, the problem was posed as a sparse error correction and solved via hard thresholding pursuit. Other approaches for solving unlabeled sensing problems are based on geometric reconstruction \cite{elhami_unlabeled_2017}, observation-specific offsets and $\ell_1$-penalization \cite{wang_estimation_2020}, or graph matching algorithms \cite{abbasi_r-local_2021}.
More recently, the problem of sample mismatches has been considered for principal component analysis \cite{yao_unlabeled_2021}, where the authors showed that it is well defined and proposed a two-stage solver based on algebraic geometry. In \cite{xu_uncoupled_2019}, the authors addressed the problem of learning models that are not necessarily linear by utilizing pairwise comparison data.

\paragraph*{Continuous-Time Sparse Signal Recovery}
In \cite{vetterli_sampling_2002}, Vetterli \textit{et al.} presented a sampling theorem for a class of sparse signals, e.g., streams of (differentiated) Diracs, piecewise polynomials or nonuniform splines, and they showed that such signals can be perfectly reconstructed from lowpass (LP) filtered samples for some families of sampling kernels, such as ideal LP filters and Gaussian kernels. Later, the class of sampling kernels allowing perfect reconstruction was extended to kernels satisfying the so-called Strang-Fix conditions \cite{dragotti_exact_2005} and the theory was adapted to also include piecewise sinusoidal signals  \cite{dragotti_sampling_2007} and handle noisy observations \cite{maravic_sampling_2005}. More recently, the conditions for perfect reconstruction for noiseless samples and sampling kernels that satisfy the Strang-Fix conditions was generalized \cite{haro_sampling_2018} and novel denoising algorithms based on structured low-rank matrix approximation were developed \cite{condat-hirabyashi-2015,haro_sampling_2018,simeoni_cpgd_2021}. Alternative methods based on convex optimization and atomic norm minimization \cite{candes-fernandez-2012,tang-etal-2013} were also proposed for the estimation of sparse signals in {\em off-the-grid} settings.

\section{Cross-Channel Unlabeled Sensing}
\label{sec:sus}
Consider a multi-channel signal $\bm x_m\in\mathbb{C}^N$, $m=1,\dots,M$, where $M$ is the number of channels and each channel signal $\bm x_m$ lies in the column space $\Es \subset \mathbb{C}^N$ of some fixed sensing matrix $\bm E\in\mathbb{C}^{N\times K}$, $N\geq K$; i.e., $\bm x_m = \bm E\bm \beta_m$ for some $\bm \beta_m\in\mathbb{C}^K$. Our interest concerns signals $\bm y_m\in\mathbb{C}^N$, where every row of the $N\times M$ matrix $\bm Y = (\bm y_1 \cdots \bm y_M)$ equals the corresponding row of the $N\times M$ matrix $\bm X =(\bm x_1 \cdots \bm x_M)$ up to a permutation of its entries, as illustrated in Fig.~\ref{fig:shuffled}. 
\begin{figure}
    \centering
    \begin{overpic}[width=0.3\textwidth,tics=5]{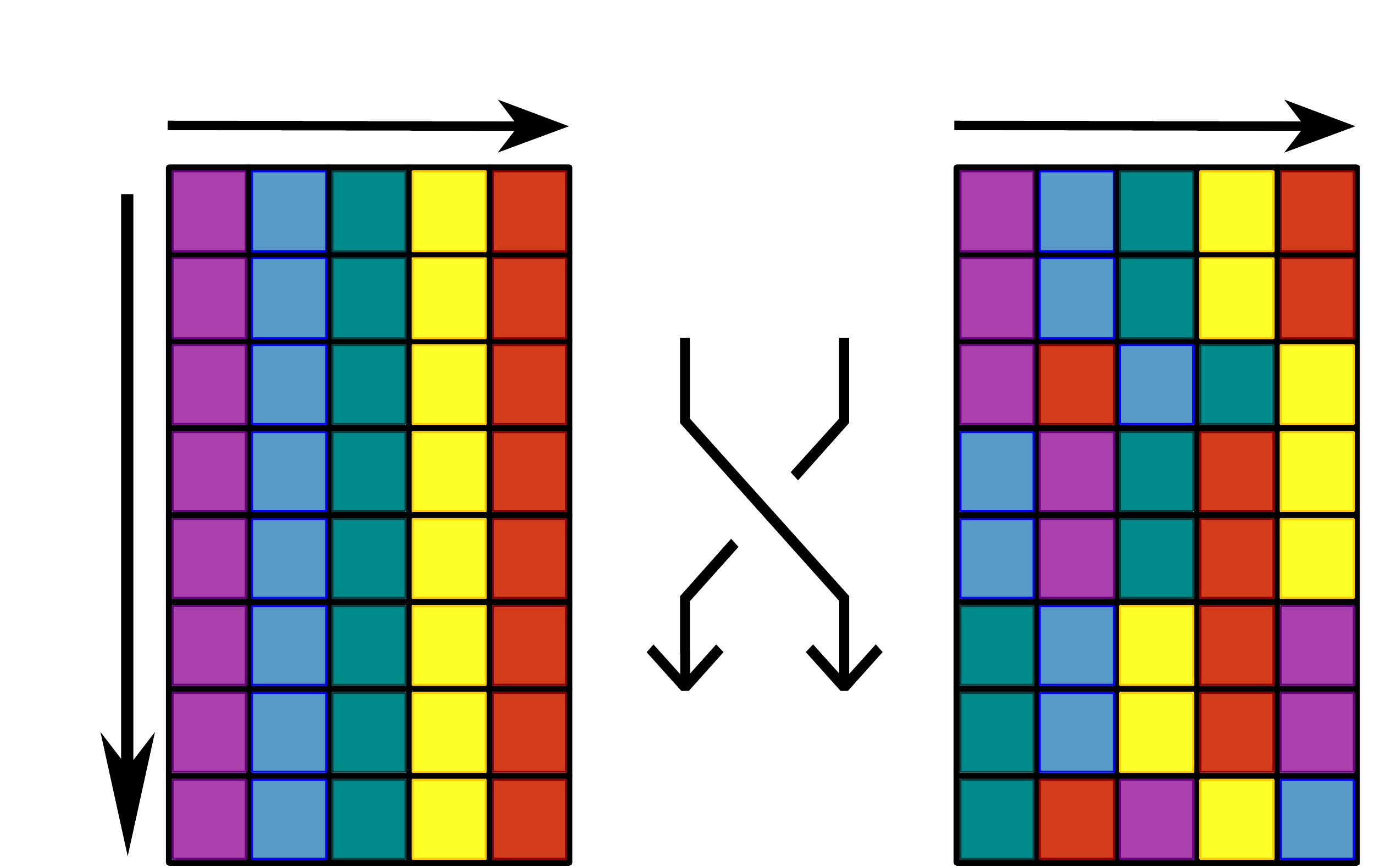}
    \put (24,57)  {$\bm x_m$}
    \put (81,57.3) {$\bm y_m$}
    \put (3,25) {$n$}
    \end{overpic}
    \caption{\textbf{Shuffled multi-channel signals:} If $\boldsymbol{x}_m$ represent the columns in a data matrix the shuffling corresponds to a permutation of each individual row.}
    \label{fig:shuffled}
\end{figure}
In other words, we are interested in multi-channel signals with imprecise correspondences between their samples and the channels. We refer to $\bm y_1, \dots, \bm y_M$ as a shuffled multi-channel signal with respect to $\bm x_1,\dots,\:\bm x_M$, and pose the question under which conditions unique recovery of all $\bm \beta_m$ is possible, when only the shuffled signals $\bm y_1,\dots,\bm y_M$ and the sensing matrix $\bm E$ are accessible, but not $\bm x_1,\dots,\:\bm x_M$. Let us first formalize our definition of shuffled signals in precise terms:
\begin{definition}[\textbf{Shuffled Multi-Channel Signal}]
\label{def:shuffled_multi_channel}
     For every pair of indices $m,n\in\{1,\dots,M\}$ we consider binary column vectors $\boldsymbol{q}_{mn}\in\{0,1\}^N$,\:  with the property 
\begin{equation}
\label{eq:sum_q}
\sum_{m=1}^M\boldsymbol{q}_{mn}~=~\sum_{n=1}^M\boldsymbol{q}_{mn}=
    \underbrace{\begin{pmatrix}
        1&1&\cdots&1
    \end{pmatrix}^\top}_{\coloneqq\boldsymbol{1}}.
\end{equation}
We call $\bm y_m\in\mathbb{C}^N, \, m=1,\dots,M$, a shuffled multi-channel signal, if it is related to the underlying multi-channel signal $\bm x_m~\in~\mathbb{C}^N, m=1,\dots,M$ by
\begin{align}
\label{eq:permutation_model_msignals}
    \underbrace{\begin{pmatrix}
        \boldsymbol{y}_{1} \\
        \boldsymbol{y}_{2} \\
        \vdots \\
        \boldsymbol{y}_{M} \\
    \end{pmatrix}}_{\eqqcolon\bm y}
    =
    \underbrace{\begin{pmatrix}
        \boldsymbol{Q}_{11} & \boldsymbol{Q}_{12}&\dots&\boldsymbol{Q}_{1M} \\
        \boldsymbol{Q}_{21} & \boldsymbol{Q}_{22}&\dots&\boldsymbol{Q}_{2M} \\
        \vdots & \vdots&\ddots&\vdots \\
        \boldsymbol{Q}_{M1} & \boldsymbol{Q}_{M2}&\dots&\boldsymbol{Q}_{MM} \\
    \end{pmatrix}}_{\eqqcolon \bm \Pi }
    \underbrace{\begin{pmatrix}
        \boldsymbol{x}_{1} \\
        \boldsymbol{x}_{2} \\
        \vdots \\
        \boldsymbol{x}_{M} \\
    \end{pmatrix}}_{\eqqcolon\bm x}, 
\end{align}
where $\boldsymbol{Q}_{mn} \coloneqq \mathrm{diag}(\boldsymbol{q}_{mn})$. 
\end{definition}
Denote by $\Qb$ the $3$-dimensional tensor whose $(m,n,p)$ entry is the $p$th entry of $\bm q_{m n}$, i.e., $\Qb_{mnp} = \bm q_{mn}(p)$. For $\bm \Pi$ as in \eqref{eq:permutation_model_msignals}, we write $\bm\Pi_{\Qb} \coloneqq \bm \Pi$ to imply its relation to $\Qb$ and $\bm\Pi_{\Qb}\coloneqq (\bm Q_{m n})$ to indicate the partition of $\bm \Pi_{\Qb}$ into $N\times N$ matrices $\bm Q_{m n}$.
As we will see in the following two lemmas, the matrices $\bm\Pi_{\Qb}$ form a subgroup of the group of permutation matrices and are hence closed under multiplication and inversion.
\begin{lemma}
\label{lem:perm_mat}
A matrix $\boldsymbol{\Pi}_\Qb$, which satisfies Definition~\ref{def:shuffled_multi_channel}, is a permutation matrix.
\end{lemma}
\begin{proof}
A permutation matrix is characterized by the property of having exactly one $1$ at each row and column, while the rest entries are zero. This is true for $\bm \Pi_\Qb$ by construction.
\end{proof}
\begin{lemma}
    \label{lem:pi1pi2}
    The set of all $\bm \Pi_\Qb$ that follow Definition~\ref{def:shuffled_multi_channel} is a subgroup of the group of permutation matrices.
\end{lemma}
\begin{proof}
We first show closure under multiplication. Let $\bm \Pi_{\Qb}$ and $\bm \Pi_{\Qb'}$ be matrices as in Definition~\ref{def:shuffled_multi_channel}. By Lemma~\ref{lem:perm_mat} these are permutation matrices, and thus so is $\bm \Pi := \bm \Pi_{\Qb} \bm \Pi_{\Qb'}$. With $\bm \Pi_{\Qb} = (\bm Q_{m n})$ and $\bm \Pi_{\Qb'} = (\bm Q'_{m n})$ as in Definition~\ref{def:shuffled_multi_channel},  it is clear that $\bm \Pi = (\bm Q''_{m n} = \sum_{\ell=1}^M \bm Q_{m \ell} \bm Q'_{\ell n})$. Since the $\bm Q_{m \ell}$'s and $\bm Q'_{\ell n}$'s are diagonal, so are the $\bm Q''_{m n}$'s, and also 
$$\sum_{n=1}^M \bm Q''_{m n} = \sum_{n=1}^M \sum_{\ell=1}^M \bm Q_{m \ell} \bm Q'_{\ell n}  = \sum_{\ell=1}^M \bm Q_{m \ell} \left(\sum_{n=1}^M \bm Q'_{\ell n}\right)=\bm I,$$
and similarly $\sum_{m=1}^M \bm Q''_{m n}=\bm I$. Hence, $\bm \Pi = \bm \Pi_{\Qb''}$ for a suitable $\Qb''$ as in Definition~\ref{def:shuffled_multi_channel}. Finally, for closure under inversion, if $\bm \Pi_{\Qb}$ is as in Definition~\ref{def:shuffled_multi_channel}, Lemma~\ref{def:shuffled_multi_channel} gives $(\bm \Pi_{\Qb})^{-1} = \bm \Pi_{\Qb}^\top = \bm \Pi_{\Qb'}$, with $\Qb'_{m n p} = \Qb_{n m p}$.  
\end{proof}

\paragraph*{Two-Channel} To give a better intuition for the problem, we now turn our attention to the case of two channels, and return to the multi-channel case subsequently. For $M=2$, relation \eqref{eq:permutation_model_msignals} reduces to
\begin{equation}
\label{eq:block_shuffled}
    \begin{pmatrix}
        \boldsymbol{y}_{1} \\
        \boldsymbol{y}_{2} \\
    \end{pmatrix}\,=\,
       \underbrace{\begin{pmatrix}
        \boldsymbol{Q}_{11} & \boldsymbol{Q}_{12}\\
        \boldsymbol{Q}_{21} & \boldsymbol{Q}_{22}\\
    \end{pmatrix}}_{\bm\Pi_\Qb}
    \begin{pmatrix}
        \boldsymbol{x}_{1} \\
        \boldsymbol{x}_{2} \\
    \end{pmatrix}
\end{equation}
with $\boldsymbol{Q}_{11} = \boldsymbol{Q}_{22} =: \bm Q$ and $\boldsymbol{Q}_{12} = \boldsymbol{Q}_{21} = \bm I -\bm Q$. We have $\bm Q = \operatorname{diag}(\bm q)$ for some $\bm q\in\{0,1\}^N$ whose zero and nonzero elements indicate whether or not measurements have been shuffled between the two channels; for simplicity we write $\bm \Pi_\Qb = \bm \Pi_{\bm q}$. 
\noindent Since $\bm x_1, \bm x_2 \in\Es$, we can expand \eqref{eq:block_shuffled} into
\begin{equation}
\label{eq:ULS_FRI_model_beta}
     \underbrace{\begin{pmatrix}
            \boldsymbol{y}_{1} \\
            \boldsymbol{y}_{2} \\
    \end{pmatrix}}_{\coloneqq\boldsymbol{y}}
        \,=\,
   \underbrace{\begin{pmatrix}
        \boldsymbol{Q} & \boldsymbol{I}\,-\,\boldsymbol{Q}\\
        \boldsymbol{I}\,-\,\boldsymbol{Q} & \boldsymbol{Q}\\
    \end{pmatrix}}_{\coloneqq\boldsymbol{\Pi}_{\bm q}}
    \underbrace{\begin{pmatrix}
        \boldsymbol{E} & \\
         & \boldsymbol{E}\\
    \end{pmatrix}}_{\coloneqq\boldsymbol{A}}
    \underbrace{\begin{pmatrix}
        \boldsymbol{\beta_1} \\
        \boldsymbol{\beta_2} \\
    \end{pmatrix}}_{\coloneqq\boldsymbol{\beta}}.
\end{equation}
In our setting, the permutation matrix $\bm \Pi_{\bm q}$ is unknown and so are the regression coefficients $\bm \beta$; this turns \eqref{eq:ULS_FRI_model_beta} into an unlabeled sensing formulation \cite{unnikrishnan_unlabeled_2018}.
Observe, however, that the sensing matrix $\bm A$ in \eqref{eq:ULS_FRI_model_beta} is highly structured: it is block diagonal and each block contains a copy of the same matrix $\bm E$. On the other hand, existing unlabeled sensing theory that studies unique recovery guarantees has so far concentrated on the case where the sensing matrix $\bm A$ is generic \cite{unnikrishnan_unlabeled_2018},\cite{dokmanic_permutations_2019}, \cite{tsakiris_algebraic-geometric_2020,tsakiris_homomorphic_2019,peng2021homomorphic,peng_homomorphic_2021,tsakiris2023determinantal}; as such, these results are not directly applicable here. Instead, we will be deriving conditions under which structured unlabeled sensing problems of the form \eqref{eq:ULS_FRI_model_beta} admit a unique or finitely many solutions. 
For that purpose, let us look at the situation where two different solutions (i.e., different pairs of permutations and regression coefficients) could give rise to the same set of observations. 
\subsubsection*{Uniqueness of the solution}
Suppose there are two-channel signal configurations $\bm x_1' = \bm E \bm \beta_1', \, \bm x_2' = \bm E \bm \beta_2'$ and $\bm x_1'' = \bm E \bm \beta_1'', \, \bm x_2'' = \bm E \bm \beta_2''$, both of which explain the observed channel data $\bm y_1, \bm y_2$. That is, there exist two permutation matrices $\bm \Pi_{\bm q'},\bm \Pi_{\bm q''}$ structured as in \eqref{eq:block_shuffled} with binary diagonal matrices $\bm Q' = \diag{(\bm q')}, \bm Q'' = \diag{(\bm q'')}$, under which $\bm \beta_1',\bm \beta_2'$ and $\bm \beta_1'',\bm \beta_2''$ are solutions to \eqref{eq:ULS_FRI_model_beta}. This gives a relation   
\begin{align}
\bm \Pi_{\bm q'''} \begin{pmatrix}
        \boldsymbol{E} & \\
         & \boldsymbol{E}\\
    \end{pmatrix}\begin{pmatrix} \bm \beta_1'' \\ \bm \beta_2'' \end{pmatrix} = 
    \begin{pmatrix}
        \boldsymbol{E} & \\
         & \boldsymbol{E}\\
    \end{pmatrix}\begin{pmatrix} \bm \beta_1' \\ \bm \beta_2' \end{pmatrix},\label{eq:id,tau}
\end{align}  where $\bm \Pi_{\bm q'''} = \bm \Pi_{\bm q'}^\top \bm \Pi_{\bm q''}$ is structured as in \eqref{eq:block_shuffled} with $\bm Q''' =~\bm I - \bm Q''-\bm Q'+2\bm Q''\bm Q'$. It thus suffices to only consider \eqref{eq:id,tau}, and for simplicity we rename $\bm q''', \bm Q'''$ to $\bm q, \bm Q$. 

Denote by $\varrho: \Ce^N \rightarrow \Ce^N$ the linear transformation given by multiplication with $\bm Q$; this is just a coordinate projection that preserves those entries of $\bm \zeta \in \Ce^N$, which correspond to the nonzero entries of $\bm q$, and sets the rest to zero. Thus, if $\bm q$ has $r$ nonzero elements, we will say that $\varrho$ preserves $r$ coordinates.
Denoting by $\varrho^\perp$ the linear transformation given by $\bm I-\bm Q$, relation \eqref{eq:id,tau} is equivalent to 
\begin{align} 
\bm x_1' & = \varrho(\bm x_1'') + \varrho^\perp(\bm x_2'') \label{eq:v_1} \\
\bm x_2' & = \varrho^\perp(\bm x_1'') + \varrho(\bm x_2''). \label{eq:v_2} 
\end{align} 
Of course, for any $\bm \zeta \in \Ce^N$, we have that $\bm \zeta = \varrho(\bm \zeta) + \varrho^\perp(\bm \zeta)$, so that the above relations are equivalent to 
\begin{align*}
    \varrho(\bm x_1'-\bm x_1'') &= \varrho^\perp(\bm x_2'' - \bm x_1')\\\varrho^\perp(\bm x_2'-\bm x_1'') &= \varrho(\bm x_2''-\bm x_2').
\end{align*}
Since $\im(\varrho) \cap \im(\varrho^\perp) = 0$ and $\bm x_1', \bm x_2', \bm x_1'', \bm x_2''$ are all in the column space $\Es$ of $\bm E$, these relations directly imply that
\begin{align}
\bm x_1' - \bm x_1'' & \in \ker(\varrho) \cap \Es, \label{eq:v_1-v_1'} \\
\bm x_1' - \bm x_2'' & \in \ker(\varrho^\perp) \cap \Es, \label{eq:v_1-v_2'} \\
\bm x_2'-\bm x_1'' & \in \ker(\varrho^\perp) \cap \Es, \label{eq:v_2-v_1'} \\
\bm x_2' - \bm x_2'' & \in \ker(\varrho) \cap \Es, \label{eq:v_2-v_2'}
\end{align} where $\ker(\varrho)$ is the kernel of $\varrho$, the set of all $\bm \zeta$'s such that $\varrho(\bm \zeta)=0$. Next, let us define what we will henceforth refer to as generic subspaces and generic matrices, respectively.

\begin{definition} [\bf{Generic Subspace, Generic Matrix}]
\label{dfn:generic}
Let $\Es \subset \Ce^N$ be a linear subspace of dimension $K$. We call $\Es$ generic if $\dim \vartheta(\Es) =K$ for every coordinate projection $\vartheta: \Ce^N \rightarrow \Ce^N$ that preserves at least $K$ entries.
We call a tall $N \times K$ matrix $\bm E$ generic, if its column space is generic, i.e., if every $K \times K$ submatrix of $\bm E$ has rank $K$. 
\end{definition}

With the above, we can show the following result:

\begin{prp} \label{prp:generic}
For a generic matrix $\bm E\in\mathbb{C}^{N\times K}$ and some binary diagonal matrix $\bm Q$ with $r$ nonzero entries, the relation 
\begin{align*}
\begin{pmatrix*}
    \bm Q & \bm I-\bm Q\\
    \bm I-\bm Q & \bm Q
\end{pmatrix*}
\begin{pmatrix}
        \boldsymbol{E} & \\
         & \boldsymbol{E}\\
    \end{pmatrix}\begin{pmatrix} \bm \beta_1'' \\ \bm \beta_2'' \end{pmatrix} = \begin{pmatrix}
        \boldsymbol{E} & \\
         & \boldsymbol{E}\\
    \end{pmatrix}\begin{pmatrix} \bm \beta_1' \\ \bm \beta_2' \end{pmatrix}
\end{align*} implies $\bm\beta_1' = \bm\beta_1''$ and $\bm\beta_2' = \bm\beta_2''$, or $\bm\beta_1' = \bm\beta_2''$ and $\bm\beta_2' = \bm\beta_1''$, if and only if $K \le \max\{r, N - r \}$. 
\end{prp}

\begin{proof}
Suppose $K \le \max\{r, N - r \}$ and \eqref{eq:id,tau} holds. Then, \eqref{eq:v_1-v_1'}-\eqref{eq:v_2-v_2'} hold. 

Suppose first that $K \le r$. Let $\mathscr{N}=\{0,\dots,N-1\}$ and let $\Ic=\{i_1 < i_2 < \cdots < i_r\}\subset\mathscr{N}$ be the coordinates that are preserved by $\varrho$ and denote by $\bm E_\Ic$ the row submatrix of $\bm E$ obtained by keeping only those rows indexed by $\Ic$. Similarly, for any vector $\bm \zeta \in \mathbb{C}^N$, denote by ${\bm \zeta}_\Ic$ the vector of elements of $\bm\zeta$ corresponding to coordinates in $\Ic$.

Since $\bm x_1' - \bm x_1'' \in \Es$, we have that $\bm x_1'-\bm x_1''  = \bm E \bm \beta$, for some $\bm \beta \in \Ce^K$. Since $\bm x_1' - \bm x_1'' \in \ker(\varrho)$, we have that $\bm 0=\varrho(\bm x_1'- \bm x_1'') $, which implies $\bm 0 = (\bm x_1'- \bm x_1'')_\Ic = \bm E_\Ic \bm\beta$. Thus, $\bm\beta$ is a right nullvector of $\bm E_\Ic$. But $\bm E_\Ic$ has full column rank $K$, because $\bm E$ is generic (Definition \ref{dfn:generic}) and $K \le r$ by hypothesis. Hence, $\bm\beta=\bm 0$ and $\bm x_1' = \bm x_1''$. Similarly, $\bm x_2'-\bm x_2'' \in \Es \cap \ker(\varrho)$ implies that $\bm x_2' = \bm x_2''$.

If instead $K \le N-r$, we let $\Ic^c = \mathscr{N} \setminus \Ic$ be the complement of $\Ic$ in $\mathscr{N}$. As before, since $\bm x_1' - \bm x_2'' \in \Es$, we have $\bm x_1' - \bm x_2'' = \bm E\bm\beta$, for some $\bm\beta \in \Ce^K$. But also $\bm x_1' - \bm x_2'' \in \ker(\varrho^\perp)$, that is $\bm 0 = \varrho^\perp(\bm x_1' - \bm x_2'')$, which implies  $\bm 0 = (\bm x_1' - \bm x_2'')_{\Ic^c} = \bm E_{\Ic^c} \bm\beta$. Since $\bm E$ is generic and $K \le N-r$, we have that $\bm\beta$ must be zero, which gives $\bm x_1' = \bm x_2''$. Arguing similarly for $\bm x_2' - \bm x_1'' \in \Es \cap \ker(\varrho^\perp)$, gives $\bm x_2' = \bm x_1''$. 

Conversely, suppose that $K > \max\{r, N - r \}$. Then, $\Es$ intersects both $\ker(\varrho)$ and $\ker(\varrho^\perp)$. Let $\bm 0 \neq \bm x_1'' \in \ker(\varrho) \cap \Es$ and $\bm 0 \neq \bm x_2'' \in \ker(\varrho^\perp) \cap \Es$. Set $\bm x_1' = \bm 0$. Then, $\bm x_1' = \varrho(\bm x_1'') + \varrho^\perp(\bm x_2'')$. Set $\bm x_2' = \bm x_1'' + \bm x_2'' \in \Es$. Since $\ker(\varrho) = \im(\varrho^\perp)$ and $\ker(\varrho^\perp) = \im(\varrho)$, we have that $\bm x_1'' \in \im(\varrho^\perp)$ and $\bm x_2'' \in \im(\varrho)$. Hence, $\varrho^\perp(\bm x_1'') = \bm x_1''$ and $\varrho(\bm x_2'') = \bm x_2''$. Thus, $\bm x_2' = \varrho^\perp(\bm x_1'')+ \varrho(\bm x_2'')$. We conclude that a set of relations of the form \eqref{eq:v_1}-\eqref{eq:v_2} hold, but $\bm 0=\bm x_1' \neq \bm x_1''$ and $\bm 0=\bm x_1' \neq \bm x_2''$. 
\end{proof}

\paragraph*{Multi-Channel} Let us now turn our attention back to the case of $M$ channels, as in Definition~\ref{def:shuffled_multi_channel}. 
Similar to the two-channel case, we are now interested in relations of the type
\begin{align}
\label{eq:observed_data}
\begin{pmatrix}
    \bm y_1\\
    \bm y_2\\
    \vdots\\
    \bm y_M\\    
\end{pmatrix}=\bm \Pi_{\Qb}
\underbrace{\begin{pmatrix}
    \bm E & &&\\
    &\bm E&&\\
    &&\ddots&\\
    &&&\bm E
\end{pmatrix}}_{\displaystyle=:\bm E^{(M)}}
\underbrace{\begin{pmatrix}
    \bm \beta_1\\
    \bm \beta_2\\
    \vdots\\
    \bm \beta_M\\
\end{pmatrix}}_{\displaystyle=:\bm \beta },
\end{align}  that link the shuffled $M$-channel signal $\bm y_1,\dots, \bm y_M$ and sensing matrix $\bm E$\textemdash{which are known\textemdash}with the unknown regression vectors $\bm \beta_1, \dots, \bm \beta_M$.

Consider $M$-channel signal configurations $\bm x_1' = \bm E \bm \beta_1',\dots,\bm x_M' = \bm E \bm \beta_M'$ and $\bm x_1'' = \bm E \bm \beta_1'',\dots,\bm x_M'' = \bm E \bm \beta_M''$, that both explain the observed channel data $\bm y_1, \dots,\bm y_M$. That is, there exist two structured permutation matrices $\bm \Pi_{\Qb'} = (\bm Q_{mn}'),\, \bm \Pi_{\Qb''}=(\bm Q_{mn}'')$ as in \eqref{eq:permutation_model_msignals}, such that 
\begin{align*}
\bm \Pi_{\Qb''} \bm E^{(M)} \bm \beta'' = \bm \Pi_{\Qb'} \bm E^{(M)} \bm \beta'.
\end{align*}
\noindent In view of Lemma~\ref{lem:pi1pi2}, it is enough to only consider 
\begin{align*}
\bm \Pi_{\Qb'''} \bm E^{(M)} \bm \beta'' = \bm E^{(M)} \bm \beta',
\end{align*}
\noindent where $\bm \Pi_{q'''}=\bm\Pi_{\Qb'}^\top\bm \Pi_{\Qb''} = (\bm Q'''_{mn}= \sum_{\ell=1}^M \boldsymbol{Q}'_{\ell m}\boldsymbol{Q}''_{\ell n})$. As with Proposition~\ref{prp:generic} for two channels, we now obtain the main result of this section, which identifies a sufficient condition for unique recovery up to a renaming of the channels:
\begin{theorem}[\bf{Cross-Channel Unlabeled Sensing}]
\label{thm:SULS}
Consider $M$-channel signal configurations $\bm x_m'=\bm E\bm \beta_m',\:\bm\beta_m'\in\mathbb{C}^K$ and $\:\bm x_m''=\bm E\bm \beta_m'',\:\bm\beta_m''\in\mathbb{C}^K,\:m=1,\dots,M$, where $\bm E\in\mathbb{C}^{N\times K}$ is generic, and binary diagonal matrices $\{\bm Q_{m n}\}_{m,n=1}^M$, satisfying $\bm I~=~\sum_{m=1}^M\bm Q_{mn}=~\sum_{n=1}^M\bm Q_{mn}$. 
Suppose that $N \ge M K$ and $\bm \beta_\kappa' \neq \bm \beta_\lambda'$ for every $\kappa \neq \lambda$.
Then, the relation
\begin{align}
\begin{pmatrix}
        \boldsymbol{Q}_{11} & \dots&\boldsymbol{Q}_{1M} \\
        \vdots & \ddots&\vdots \\
        \boldsymbol{Q}_{M1} &\dots&\boldsymbol{Q}_{MM} \\
    \end{pmatrix}
    \begin{pmatrix}
    \bm x_1''\\
    \vdots\\
    \bm x_M''\\
\end{pmatrix} = 
\begin{pmatrix}
    \bm x_1'\\
    \vdots\\
    \bm x_M'\\
\end{pmatrix} \label{eq:M-channel-US}
\end{align}
implies $\bm x_m' = \bm x_{n_m}''$ for every $m$, with $\{n_m\}_{m=1}^M = \{m\}_{m=1}^M$. 
\end{theorem}
\begin{proof}
Let $\rho_{mn}: \mathbb{C}^N \rightarrow \mathbb{C}^N$ be the coordinate projection given by multiplication with $\bm Q_{mn}$, and denote by $r_{mn}$ the number of nonzero elements in $\bm Q_{mn}$. Then, we can write the $m$th row of the matrix equation \eqref{eq:M-channel-US} as 
\begin{align}
\bm x_m' = \rho_{m1}(\bm x_1'') + \cdots + \rho_{mM}(\bm x_M'').  \label{eq:M-channel-US-ith-row}
\end{align} The $\rho_{mn}$'s with $m$ fixed form an \emph{orthogonal resolution of the identity} \cite{Roman}, i.e.: i) $\rho_{m1} + \cdots + \rho_{mM}$ is the identity map $\mathbb{C}^N \rightarrow \mathbb{C}^N$, ii) $\rho_{mn} \rho_{mn'} = \rho_{mn'} \rho_{mn}=0$ for $n \neq n'$, and iii) each $\rho_{mn}$ is an orthogonal projection. Hence, \eqref{eq:M-channel-US-ith-row} yields
\begin{align}
\sum_{n=1}^M \rho_{mn}( \bm x_m' - \bm x_n'' )=0. \label{eq:M-channel-US-ith-row-decomposed}
\end{align} Since $N \ge M K$, it cannot be that $r_{mn} < K$ for every $m$. Thus, there exists some $n_m \in\{1,\dots,M\}$, such that $r_{m n_m} \ge K$. Applying $\rho_{m n_m}$ to both sides of \eqref{eq:M-channel-US-ith-row-decomposed}, we get 
$$ \rho_{m n_m}(\bm x_m' - \bm x_{n_m}'')=0.$$ In turn, this gives us that $\bm x_m' - \bm x_{n_m}'' \in \mathscr{E} \cap \ker (\rho_{m n_m})$, where $\mathscr{E}$ is the column space of $\bm E$, and as in the proof of Proposition~\ref{prp:generic}, we obtain $\bm x_m' = \bm x_{n_m}''$. There was nothing special about our choice of $m$, and so for every $m = 1,\dots, M$ there is an $n_m$ such that $\bm x_m' = \bm x_{n_m}''$. If $n_{\kappa} = n_{\lambda}$ for some $\kappa \neq \lambda$, then $$\bm x_\kappa' =  \bm x_{n_\kappa}'' = \bm x_{n_\lambda}''  = \bm x_{\lambda}',$$ whence $\bm E (\bm \beta_\kappa - \bm \beta_\lambda) = \bm 0$. Since $\bm E$ is generic, it has full column rank, by which we arrive at the contradiction $\bm \beta_\kappa = \bm \beta_\lambda$. We conclude that all $n_m$'s are distinct.
\end{proof}

\section{Recovery of Shuffled Sparse Signals}
\label{sec:sss}
Theorem~\ref{thm:SULS} establishes conditions on the dimensions and number of shuffled samples that allow for unique recovery in the cross-channel unlabeled sensing problem, assuming the underlying signals lie in a generic subspace. In this section, we extend the results of Theorem~\ref{thm:SULS} and link them to the continuous-time sparse signal reconstruction problem. This leads to a cross-channel unlabeled sensing problem where the sensing matrix is not known precisely. We show that, assuming the signals of interest admit a sparse representation in an overcomplete dictionary, it is possible to decouple the sensing matrix recovery and the unlabeled sensing problem.
\subsection{Sparse Signal Model}\label{ssec:signal}
We consider continuous-time sparse signals that can be characterized by a weighted sum of spikes supported on the unit interval. Such signals may be, for instance, the result of temporal point processes, which are widely used stochastic models that approximate, e.g., the firing activity of neurons \cite{Eden2016} or wireless network signals \cite{keeler2018wireless}. Formally, we write:
\begin{equation}\label{eq:stream_diracs}
    x(t) = \sum_{k} a_k\,\delta(t-t_k),
\end{equation}
where $\delta(\cdot)$ is the Dirac delta function\footnote{$\int_{-\infty}^{\infty} f(t)\delta(t-t_0)\,dt = f(t_0),\quad \int_{-\infty}^{\infty}\delta(t)\,dt = 1$}, and $t_k\in[0,1)$ and $a_k\in\mathbb{C}$ are the respective locations and weights of the spikes. 
Note that $x(t)$ in \eqref{eq:stream_diracs} has {\em Fourier Series} (FS) coefficients 
\begin{equation}\label{eq:fourier_series}
    X_{\ell} = \int_0^1 x(t)\,e^{j2\pi t\ell}\, dt= \sum_ka_k\,e^{-j2\pi t_k\ell },\quad \ell\in\mathbb{Z}.
\end{equation} 
In a sampling setup as the one depicted in Fig.~\ref{fig:shuffling}, and given an appropriate choice of the sampling kernel $\varphi(\cdot)$, it is possible to obtain samples $x_n = \langle x(t),\varphi(t-n/T_s)\rangle$ whose Discrete Fourier Transform (DFT) corresponds to the true FS coefficients $X_\ell$ of the signal \cite{vetterli_sampling_2002,dragotti_sampling_2007,bejar_finite_2020}. Therefore, in our multi-channel sampling setup we consider signals $\bm x_m$ given as mixtures of complex exponentials in the frequency domain:
\begin{align}
\label{eq:Xl}
    &X_{m\ell} \coloneqq [\operatorname{DFT}_N(\bm x_{m})]_\ell = \sum_{k=0}^{K_m-1} a_{mk}\,e^{-j2\pi t_{mk}\ell},
\end{align}
where $[\operatorname{DFT}_N(\cdot)]_\ell$ denotes the $\ell$th element of the $N$-point DFT of a discrete vector signal, $K_m \in\mathbb{N}_+$ is the number of components, and where $\mathcal{T}_m \coloneqq\{t_{mk}\in [0,1)\}_{k=0}^{K_m-1}$ and $\mathcal{A}_m\coloneqq \{a_{mk}\in\mathbb{C}\}_{k=0}^{K_m-1}$ are the locations and weights of the spikes for the $m$th channel.
\subsection{Shuffled Sparse Signal Recovery}\label{ssec:shuffled_sparse}
\begin{figure*}[t]
    \centering
    \begin{overpic}[width=0.99\textwidth,tics=5]{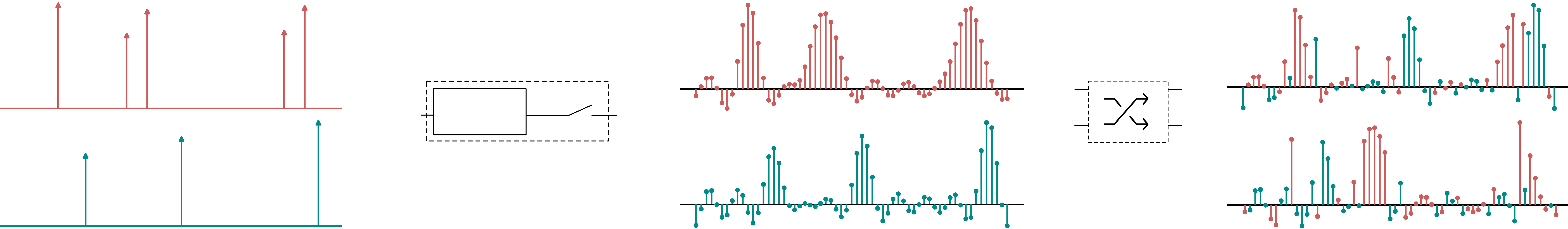}
        \put (65.8,9.75) {$x_{1n}$}
        \put (65.8,7.5) {$x_{2n}$}
        \put (74.8,9.75) {$y_{1n}$}
        \put (74.8,7.5) {$y_{2n}$}
        \put (39.5,6.75) {$x_{mn}$}
        \put (22,6.75) {$x_{m}(t)$}
        \put (27.9,7) {{\small$\varphi^*(\frac{-t}{T_s})$}}
        \put (32.75,9.75) {$\Phi^*$}
        \put (34.75,8) {\footnotesize$nT_s$}
    \end{overpic}
    \caption{\textbf{Sampling Setup - Shuffled Sparse Signals:} The measurements $\boldsymbol{y}_m$ of the $m$th channel are modeled as a shuffling of the samples in the underlying signals $\boldsymbol{x}_m$ at each sampling point $n$.}
    \label{fig:shuffling}
\end{figure*}
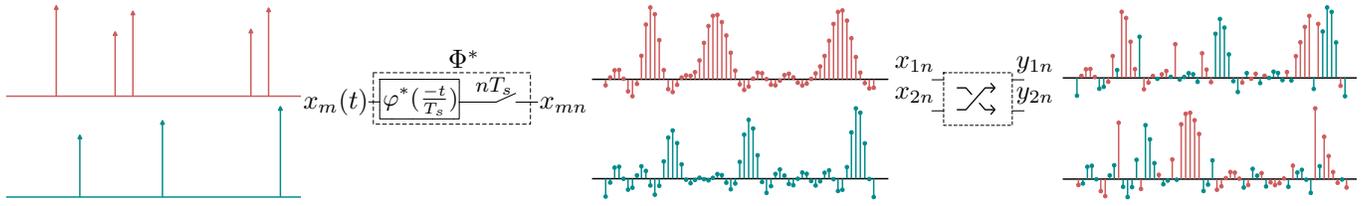
Consider now that we observe a shuffled multi-channel signal $\bm y_1,\dots,\bm y_M$ as in Definition~\ref{def:shuffled_multi_channel}, where the unknown $\bm x_m$'s are of the form of \eqref{eq:Xl}. Contrary to the previous section, the individual channels $\bm x_m = \bm E_m\bm a_m,\: m =1,\dots,M$ with $\:\bm a_m\in\Ce^{K_m}$ live in unknown, possibly different subspaces $\Es_m\subset\Ce^{N}$ spanned by the columns of unknown sensing matrices $\bm E_m\in\Ce^{N\times K_m}$.

If the $\bm x_m$'s were known, and provided the number of measurements $N$ is sufficiently large, it is possible to retrieve the parameters of the channels $\Theta_m = \{(t_{mk},a_{mk})\}_{k=0}^{K_m-1}$. In particular, it is well-known that for a uniformly sampled mixture of $K$ distinct complex exponentials as in \eqref{eq:Xl}, a number of $N=2K$ samples uniquely determines the parameters of the mixture, which can be recovered by using line spectral estimation methods, such as Prony's method \cite{prony-1795}\cite{stoica-moses-2005}. The sensing matrix $\bm E_m$ with column space $\Es_m$ is then expressed as $\bm E_m \coloneqq \bm{W V}_m$, where $\bm W$ is the $N \times N$ inverse discrete Fourier transform (IDFT) matrix and $\bm V_m$ is an $N\times K_m$ Vandermonde matrix (i.e., $[\bm V_m]_{\ell k} = v_{mk}^\ell$ $v_{mk}=\exp(-j2\pi t_{mk})$). Then, each channel can be written as $\bm x_m = \bm E_m \bm a_m$ with $\: \bm a_m = (a_{m,0},\dots,a_{m,K_m-1})^\top$.

In our case, however, we cannot directly access the FS coefficients of the channels due to the shuffling of the temporal signals $\bm x_m$. Instead, the special structure of \eqref{eq:permutation_model_msignals} allows us to look at the addition of the measured vectors since this quantity is invariant with respect to the permutation of the samples across channels, i.e.,
\begin{align}
\label{eq:add_vectors}
\boldsymbol{x}_\Sigma\,\coloneqq\,\sum_{m=1}^M\bm{x}_{m}\,&=\,\sum_{m=1}^M\boldsymbol{y}_{m}\,\eqqcolon \boldsymbol{y}_\Sigma.
\end{align}
Based on \eqref{eq:add_vectors}, $[\operatorname{DFT}_N(\bm x_{\Sigma})]_\ell$ is again a mixture of complex exponentials as defined in \eqref{eq:Xl}. Let us now denote $\Theta_\Sigma=\{(a_{\Sigma,k} \in \mathbb{C},\:t_{\Sigma,k} \in [0,1)\,)\}_{k=0}^{K_\Sigma-1}$ the parameters of the mixture $\bm x_\Sigma$, where $K_\Sigma$ is the number of distinct locations $t_{\Sigma,k}$ in $\bm x_\Sigma$ with corresponding weights $a_{\Sigma,k}$. Suppose that $\{t_{\Sigma,k}\}_{k=0}^{K_\Sigma-1}=\bigcup_{m=1}^M\mathcal{T}_m$ with nonvanishing $a_{\Sigma, k}$'s. Note that this condition only fails to hold, if multiple channels have intersecting sets of locations $\mathcal{T}_m$, and the $a_{mk}$'s associated with the same location annihilate each other. From a probabilistic point of view, however, such an event is extremely unlikely if the $t_{mk}$ were randomly sampled from the unit interval. If the number of samples now satisfies $N\geq 2K_\Sigma$, $\Theta_\Sigma$ can be uniquely recovered via line spectral estimation, and from the $t_{\Sigma,k}$'s we arrive at a sensing matrix of the form
\begin{equation}
\label{eq:U}
    \bm E \coloneqq~\bm {W} \bm V_\Sigma,
\end{equation} 
where $\bm V_\Sigma$ is an $N\times K_\Sigma$ Vandermonde matrix with $[\bm V_\Sigma]_{\ell k} = v_{\Sigma,k}^\ell$, $\{v_{\Sigma,k}=\exp(-j2\pi t_{\Sigma,k})\}_{k=0}^{K_\Sigma-1}=\bigcup_{m=1}^M \{v_{mk}\}_{k=1}^{K_m-1}$. The column space of $\bm E\in\Ce^{N\times K_\Sigma}$ is then given by $\Es = \Es_1 + \dots + \Es_M \subset \Ce^N$, and hence $\bm x_m\in\Es$ for all $m$.
The described procedure allows us to retrieve the support of the spikes which, in turn, allows us to represent each underlying channel signal as $\bm x_m = \bm E \bm\beta_m$ for some $\bm\beta_m\in\Ce^{K_\Sigma}$. Now, the problem of recovering the signals from the mismatched sample-channel assignments reduces to the cross-channel unlabeled sensing problem in \eqref{eq:ULS_FRI_model_beta}.\newline\indent
To show that unique recovery is possible by using the results developed in the previous section, it remains to be shown that the matrix $\bm E = \bm W \bm V_\Sigma$ is generic in the sense of Definition~\ref{dfn:generic}, under good circumstances. This is done in the next lemma:

\begin{lemma}
\label{lem:Cauchy}
Let $\bm W$ be the $N \times N$ IDFT matrix and $\bm V$ an $N \times K$ Vandermonde matrix with $[\bm V]_{\ell k}=v_k^{\ell}$, such that the $v_k\in\Ce$ are distinct and $ v_k^N \neq \pm  1$ for all $k$. Then, the matrix $\bm W \bm V$ is generic in the sense of Definition~\ref{dfn:generic}. 
\end{lemma}
\begin{proof}
Without loss of generality, let $\bm W$ be the scaled $N \times N$ IDFT matrix with $[\bm W]_{i\ell}= \omega^{i\ell}$ and $\omega = \exp(j2\pi/N)$. Let $\bm V$ be the $N \times K$ Vandermonde matrix with $[\bm V]_{\ell k} = v_{k}^{\ell}$ (note that the indices $\ell, k$ and $i$ start at $0$).\newline\indent
Let $\Lc = \{\ell_0 < \cdots < \ell_{K-1} \}$ be any subset of $\mathscr{N}=\{1,\dots,N\}$ of cardinality $|\Lc|=K$. We will show that the $K \times K$ submatrix $[\bm W \bm V]_\Lc$ obtained by keeping rows with indices in ${\Lc}$ is invertible. That submatrix is clearly equal to $\bm W_\Lc \bm V$, where $\bm W_\Lc$ is the row submatrix of $\bm W$ obtained by selecting rows indexed by ${\Lc}$. 
The key observation is that the matrix $\bm W_\Lc \bm V$ can be factorized as $\bm W_\Lc \bm V =~\bm D_1 \bm C \bm D_2$, where $\bm D_1, \bm D_2 $ are diagonal and $\bm C$ is a Cauchy matrix:
 \begin{align*}
        [\bm W_\Lc \bm V]_{ik} &= \sum_{\ell=0}^{N-1}v_k^\ell \omega^{\ell_i\ell} =  \frac{1-v_k^N\omega^{N\ell_i}}{1-v_k\omega^{\ell_i}}\\&
        =\underbrace{(1-v_k^N)}_{=[\bm D_1]_{kk}}\underbrace{\frac{1}{\omega^{-\ell_i}-v_k}}_{=[\bm C]_{ik}}\underbrace{\frac{1}{\omega^{\ell_i}}.}_{=[\bm D_2]_{ii}}
    \end{align*}
Now, the determinant of a Cauchy matrix has a well-known closed form expression; from this we  find that 
\begin{align*}
\det(\bm W_\Lc\bm V)
=\medmath{\frac{\displaystyle\prod_{k=0}^{K-1}(1-v_k^N) \prod_{0\leq k' < k \leq K-1}\!\!\!\!\!\!\!\!\!(\omega^{-\ell_k}-\omega^{-\ell_{k'}})(v_{k'}-v_{k})}{\displaystyle\prod_{i = 0}^{K-1}\omega^{\ell_i}\prod_{i, k = 0}^{K-1}(\omega^{-\ell_i}-v_{k})}},
\end{align*}
which is never zero since the $v_{k}$'s are distinct and $v_{k}^N~\neq~\pm~1$.
\end{proof}

We are now ready to present our main result for shuffled multi-channel sparse signal recovery:
 \begin{theorem}
\label{theorem:1}
Consider $M$ distinct signals $\{\bm x_m \in \Ce^N\}_{m=1}^M$ with $[\operatorname{DFT}_N(\bm x_{m})]_\ell=\sum_{k=0}^{K_m-1} a_{mk}v^{\ell t_{mk}},$ $v = \exp(-j2\pi)$, where $\mathcal{T}_m=\{t_{mk}\in[0,1)\}_{k=0}^{K_m-1}$ is the set of (distinct) locations with associated weights $\mathcal{A}_m=\{a_{mk}\in\Ce\}_{k=0}^{K_m-1}$, and $v^{Nt_{mk}}\ne\pm1$ for all $k$. For $\bm x = \sum_{m=1}^M\bm x_m$, suppose that $[\operatorname{DFT}_N(\bm x)]_\ell=\sum_{k=0}^{K_m-1} a_{k}v^{\ell t_{k}}$ with locations $\{t_k\}_{k=1}^K=\bigcup_{m=1}^M\mathcal{T}_m$ and nonvanishing weights $a_k$. Further, assume that we observe the channels $\{\bm y_m\}_{m=1}^M$ of a shuffled M-channel signal w.r.t. the $\bm x_m$'s. If the number of samples $N\geq MK$, then $\{\mathcal{A}_m,\mathcal{T}_m\}_{m=1}^M$ can be uniquely recovered up to a renaming of the channels.
\end{theorem} 
\begin{proof}
From Theorem~\ref{thm:SULS}, we know that $\{\bm x_m=\bm E\bm\beta_m\}_{m=1}^M$ can be uniquely recovered from $\{\bm y_m\}_{m=1}^M$ up to a renaming of channels, if i) $N\geq MK$, ii) all $\bm x_m$'s are distinct, and $\bm E\in \Ce^{N\times K}$ is iii) known and iv) generic. i) and ii) hold by hypothesis. It thus remains to verify iii) and iv). 

\noindent With $X_{m\ell}$ as in $\eqref{eq:Xl}$, then the DFT of $\bm x = \sum_{m=1}^M \bm y_m$ is a mixture of complex exponentials. By assumption, the set of locations of such a mixture is given by $\mathcal{T}=\bigcup_{m=1}^M\mathcal{T}_m$, with the resulting amplitudes associated with each location being nonvanishing. With $K = |\mathcal{T}|$, and since $N\geq MK\geq2K$ with $t_{mk}$ distinct for all $k$ and fixed $m$, we can recover the set of locations $\mathcal{T}$ via line spectral estimation methods. Then, (iii) is satisfied by constructing a matrix $\bm E \eqqcolon \bm W\bm V$, where $\bm W$ is the $N\times N$ IDFT matrix and $\bm V$ is an $N\times K$ Vandermonde matrix with $[\bm V]_{\ell k}=v_k^\ell,\:v_{k}=\exp(-j2\pi t_{k}),$ with locations $\{t_k\}_{k=0}^{K-1}\coloneqq\mathcal{T}$. Condition iv) follows directly from Lemma~\ref{lem:Cauchy}, since $\bm V$ is Vandermonde with all $v_{k}'s$ distinct and $v_{k}^N\neq\pm1$. 
\end{proof}

\section{Proposed Two-Step Estimation Approach}
\label{ssec:rob_est}
In this section, we devise an empirical method for reconstructing shuffled sparse signals in the two-channel case. We propose a two-step estimation approach based on the results from Section~\ref{sec:sss}. Step~1 is dedicated solely to the estimation of the signals' support on the continuum from their addition, while Step~2 builds upon this estimate and attempts to solve the resulting unlabeled sensing problem. However, in a practical setting the observations are corrupted with measurement noise. Hence, the signals are modeled as the sum of a deterministic component admitting a sparse representation according to (\ref{eq:Xl}) and a random perturbation term that is assumed to be normally distributed, i.e.,  
\begin{equation}
\label{eq:noisy_unmixed}
    \Tilde{\boldsymbol{x}}_m\,=\,\boldsymbol{x}_{m}\,+\,\boldsymbol{\nu}_m,\:\boldsymbol{\nu}_m\,\sim\,\mathcal{N}(0,\sigma_m^2),\:m=1,2.
\end{equation}
\begin{remark}
Note that Theorem~\ref{theorem:1} guarantees a unique recovery of the parameters of ${\boldsymbol{x}}_m$ only if $\sigma_m=0$. Additionally, the sensing matrix has to be constructed from the estimated locations of the noisy signals $\Tilde{\boldsymbol{x}}_1,\Tilde{\boldsymbol{x}}_2$. The estimation errors from Step~$1$ are reflected in the estimated sensing matrix $\hat{\boldsymbol{A}}$ as so-called leverage points, while shuffled samples that deviate largely from each other, induce outliers in the residuals with respect to the assumed Gaussian distribution.
\end{remark}
\paragraph*{{{Step\:1 - Estimating Support}}} First, the two shuffled measurement vectors $\Tilde{\boldsymbol{y}}_1$ and $\Tilde{\boldsymbol{y}}_2$ with $\tilde{\bm y} \coloneqq (\tilde{\bm y_1}^\top,\tilde{\bm y_2}^\top)^\top$, $\tilde{\bm x} \coloneqq (\tilde{\bm x_1}^\top,\tilde{\bm x_2}^\top)^\top$ and $\tilde{\bm y} \coloneqq \bm\Pi_{\bm q}\tilde{\bm x}$ as in \eqref{eq:block_shuffled} are added, i.e., $\Tilde{\boldsymbol{y}}_{\Sigma}\,=\,\Tilde{\boldsymbol{x}}_{\Sigma}\,=\,\Tilde{\boldsymbol{x}}_1+\Tilde{\boldsymbol{x}}_2=\Tilde{\boldsymbol{y}}_1+\Tilde{\boldsymbol{y}}_2$, where $\Tilde{\boldsymbol{x}}_1$ and $\Tilde{\boldsymbol{x}}_2$ are the noisy unshuffled measurement vectors as defined in (\ref{eq:noisy_unmixed}). Thus, the sum is given by  $\Tilde{\boldsymbol{y}}_{\Sigma}\,=\,\boldsymbol{y}_{\Sigma}\,+\,\boldsymbol{\nu}_{\Sigma}$ with $\boldsymbol{\nu}_\Sigma\,\sim\,\mathcal{N}(0,\sigma_1^2\,+\,\sigma_2^2)$. 
To estimate the parameters of the signal mixture one could resort to the methods described in Section~\ref{sec:related_work}. However, many of these methods rely on the low-rank property of the Toeplitz matrix constructed from the FS coefficients of the signal of interest, which is not necessarily present in the case of noisy samples. In our case, we propose to use methods that impose a low-rank constraint on the Toeplitz matrix formed from the sequence of DFT coefficients $[\operatorname{DFT}_N(\tilde{\bm x}_{m})]_\ell$ \cite{cadzow_signal_1988}\cite{condat-hirabyashi-2015},\cite{haro_sampling_2018},\cite{simeoni_cpgd_2021}. These methods can be seen as a denoising step to get a better estimate of the true $X_{\Sigma,\ell}$ prior to the application of a line spectral estimation method (e.g., Prony's method \cite{prony-1795}) to recover the signal parameters. Here, we choose the method in \cite{haro_sampling_2018} which uses ADMM to solve the denoising problem with low-rank constraints due to its superior performance and faster convergence as compared to \cite{condat-hirabyashi-2015,simeoni_cpgd_2021}. Once we get an estimate of the signal parameters $\hat\Theta_\Sigma$ we can compute the matrix $\hat{\bm E} = \bm W\hat{\bm V}_\Sigma$ where $\hat{\bm V}_\Sigma$ is an estimate of the true $\bm V_\Sigma$.
\paragraph*{{{Step~2 - Shuffled Regression}}}
\label{ssec:step2}
While there are approaches to solve the general unlabeled sensing problem that assume a known sensing matrix \cite{tsakiris_algebraic-geometric_2020}, \cite{peng_algebraically-initialized_2019}, \cite{unnikrishnan_unlabeled_2018}, \cite{slawski_linear_2019}, \cite{abid_stochastic_2018}, \cite{slawski_pseudo-likelihood_2021}, \cite{peng_homomorphic_2021}, to the best of our knowledge, there exists no method that assumes contaminated regressors, while also taking into account the structure of the permutation and the sensing matrix. Assuming Gaussian noise, for a contamination-free regression matrix (i.e., $\hat{\bm A} = \bm A$), the maximum likelihood estimate (MLE) of the cross-channel unlabeled sensing problem is the solution to
\begin{align}
{(\hat{\boldsymbol\Pi}_{\bm q,\mathrm{ML}},\hat{\boldsymbol{\beta}}_{\mathrm{ML}}) =}\ & \underset{\boldsymbol{\Pi}_{\bm q},\boldsymbol{\beta}}{\mathrm{arg\, min}}\:\:
     ||\Tilde{\boldsymbol{y}}\,-\,
    \boldsymbol{\Pi}_{\bm q}
    \boldsymbol{A}
    \boldsymbol{\beta}
    ||_2^2\\
\mathrm{s.t.}&\:\:
{\bm q}\,\in\, \{0,\:1\}^N,\:\boldsymbol{\beta}\,\in\,\mathbb{R}^{2K_\Sigma}.\nonumber
\end{align}
However, we do not have access to $\bm A$ but rather to an estimate $\hat{\bm A}$ of it, which violates the assumptions that are required for the MLE to be optimal. Still, if $\hat{\bm A}$ is sufficiently close to $\bm A$, the MLE is still a reasonable choice to solve the following problem instead:
\begin{align}
\label{eq:obj_bilin_binary}
 \underset{\boldsymbol{\Pi}_{\bm q},\boldsymbol{\beta}}{\mathrm{min}}\:\:
     ||\Tilde{\boldsymbol{y}}\,-\,
    \boldsymbol{\Pi}_{\bm q}
    \hat{\boldsymbol{A}}
    \boldsymbol{\beta}
    ||_2^2
\quad\mathrm{s.t.}\:\:
{\bm q}\,\in\, \{0,\:1\}^N,\:\boldsymbol{\beta}\,\in\,\mathbb{R}^{2K_\Sigma}.
\end{align}
The combinatorial nature of problem \eqref{eq:obj_bilin_binary} makes it difficult to solve, and a natural approach is to relax $\boldsymbol{q}$ to lie in the convex set $[0,\:1]^N$.
Unfortunately, in any kind of alternate minimization approach, where the regression problem is solved for a previously estimated permutation matrix, the response vector may be contaminated by arbitrary outliers due to the shuffling. It is a well-known fact that these errors lead to a severe drop in performance of the MLE \cite{p_j_huber_robust_2009}, \cite{f_r_hampel_robust_2011}, \cite{ yohai_high_1987}, \cite{zoubir_robust_2018}. Robust estimators have been developed as a means to deal with such erroneous data. For the linear regression problem, $\tilde{\boldsymbol{y}}\,=\,{\bm H}{\boldsymbol{\beta}}+ \boldsymbol{\varepsilon}$ in particular, the goal is to estimate $\boldsymbol{\beta}$ even when ${\bm H}$ contains outliers and $\bm\varepsilon$ is non-Gaussian and possibly heavy-tailed. Given estimates $\hat{\bm\Pi}_{\bm q}$ and $\hat{\bm A}$, we can define ${\bm H}:=\hat{\bm\Pi}_{\bm q}\hat{\bm A}$ and solve for $\bm\beta$ using a robust estimator. In this work, the robust MM-estimator \cite{yohai_high_1987}, which uses the iteratively reweighted least-squares (IRWLS) algorithm for computation, is adopted due to its high breakdown point (maximum proportion of outlier-contamination beyond which its asymptotic bias is infinite) and high asymptotic relative efficiency (ratio of the robust estimators asymptotic variance compared to that of the Gaussian MLE for normally distributed noise) for normally distributed residuals\cite{zoubir_robust_2018}. 
\begin{algorithm}
    \caption{Shuffled Sparse Signal Recovery (SSSR)}
    \label{algo:RTSE}
    \hspace*{\algorithmicindent}\textbf{Input:}
    $\Tilde{\boldsymbol{y}}_1,\Tilde{\boldsymbol{y}}_2\in\mathbb{R}^N,\:K_\Sigma>2,\:\mathrm{maxiter}>0$
    \begin{algorithmic}[1]
        \State  Initialize: $\hat{\boldsymbol{q}}\,=\,(1,\,\dots,\,1)^\top,\:\hat{\boldsymbol{\Pi}}_q$ according to (\ref{eq:block_shuffled}), $\Tilde{\boldsymbol{y}}\,=\,(\Tilde{\boldsymbol{y}}_1^\top,\,\Tilde{\boldsymbol{y}}_2^\top)^\top, \:\hat{\boldsymbol{x}}\,=\,\hat{\boldsymbol{\Pi}}_{\bm q}^\top\Tilde{\boldsymbol{y}}$
        \State  $\hat{X}_{\Sigma,\ell} \gets$ DFT and denoise $ (\Tilde{\boldsymbol{y}}_1\,+\,\Tilde{\boldsymbol{y}}_2)$ \cite{haro_sampling_2018} 
        \State  $\{\hat{t}_k\}_{k=0}^{K_\Sigma-1}\gets$ Line spectral estimation of $\hat{X}_{\Sigma,\ell}$ \cite{prony-1795}
        \State  $\hat{\boldsymbol{A}} \gets \mathrm{diag}(\hat{\boldsymbol{E}}), \:\hat{\boldsymbol{E}}$ from $\{\hat{t}_k\}_{k=0}^{K_\Sigma-1}$ according to
        \Statex $\boldsymbol{\hat{E}}=\boldsymbol{W}\boldsymbol{\hat{V}}_\Sigma$
        \State \textbf{while} $\mathrm{iter}\,\leq\,\mathrm{maxiter}$ \textbf{do}:
        \State\hspace{\algorithmicindent}   $\hat{\boldsymbol{\beta}}_{\mathrm{rob}} \gets$ Robust estimate of linear
        \Statex\hspace{\algorithmicindent} regression $\hat{\boldsymbol{x}} =\,{\hat{\boldsymbol{A}}}{\boldsymbol{\beta}} + \boldsymbol{\varepsilon}$ \cite{yohai_high_1987}
        \State\hspace{\algorithmicindent}  $\hat{\boldsymbol{q}}\gets\mathrm{proj}_{\{0,1\}}\left(\underset{\boldsymbol{q}\,\in\,[0,\,1]^N}{\mathrm{argmin}}\:\:\|\hat{\boldsymbol{\Pi}}_{\bm q}^\top\Tilde{\boldsymbol{y}}\,-\,\hat{\boldsymbol{A}}\hat{\boldsymbol{\beta}}_{\mathrm{rob}}\|_2^2\right)$
        \State\hspace{\algorithmicindent}  $\hat{\boldsymbol{x}}\gets\hat{\boldsymbol{\Pi}}_{\bm q}^\top\Tilde{\boldsymbol{y}}$ 
        \State\hspace{\algorithmicindent} $\mathrm{iter} \gets \mathrm{iter}\,+\,1$
        \State Select $\hat{\boldsymbol{\beta}},\,\hat{\boldsymbol{q}}$ with smallest $\mathrm{MSE}\,=\,\|\hat{\boldsymbol{\Pi}}_{\bm q}^\top\Tilde{\boldsymbol{y}}\,-\,\hat{\boldsymbol{A}}\hat{\boldsymbol{\beta}}_{\mathrm{rob}}\|_2^2$
    \end{algorithmic}
    \hspace*{\algorithmicindent} \textbf{Output:} $\hat{\boldsymbol{A}},\:\hat{\boldsymbol{\beta}},\:\hat{\boldsymbol{q}}$
\end{algorithm}
After initializing $\hat{\boldsymbol{q}}\,=\,(1,\,\dots,\,1)$ and $\hat{\bm x} = \hat{\boldsymbol{\Pi}}_{\bm q}^\top\tilde{\boldsymbol{y}}$, where $\hat{\boldsymbol{\Pi}}_{\bm q}$ is constructed from $\hat{\boldsymbol{q}}$ following the definition in (\ref{eq:block_shuffled}), the robust estimate $\hat{\boldsymbol{\beta}}_{\mathrm{rob}}$ is computed w.r.t. $\hat{\bm x}$. Then, the convex program
\begin{align}
\label{eq:conv_program}
 \hat{\boldsymbol{q}}_{[0,1]}\,=\,&\underset{\boldsymbol{q}}{\mathrm{argmin}}\:\:\|\tilde{\boldsymbol{y}}\,-{\boldsymbol{\Pi}}_{\bm q}\,\hat{\boldsymbol{A}}\hat{\boldsymbol{\beta}}_{\mathrm{rob}}\|_2^2,\,\\&\mathrm{\:s.t.\:} {\bm q}\,\in\,[0,1]^N,\nonumber
\end{align}
is evaluated and the estimate $\hat{\boldsymbol{q}}_{[0,1]}$ is projected onto the binary set $\{0,\,1\}^N$, i.e., $\hat{\boldsymbol{q}}\,=\,\mathrm{proj}_{\{0,\,1\}}(\hat{\boldsymbol{q}}_{[0,\,1]})$, which serves as a proxy for the true assignment vector $\boldsymbol{q}$. Again, from $\hat{\boldsymbol{q}}$ the matrix $\hat{\boldsymbol{\Pi}}_{\bm q}$ is constructed and the measurements are permuted according to $\hat{\bm x}=\hat{\boldsymbol{\Pi}}_{\bm q}^\top\boldsymbol{y}$.
Unfortunately, iterating this approach does not guarantee convergence, due to the nonconvexity of the MM-estimator as well as the binary projection. However, since the MM-estimator has a high breakdown point, making a few mistakes in the sample assignment has a small impact on the estimated model. Therefore, the process is repeated until some predefined number of maximum iterations is reached. Then from all iterations, the assignment vector $\hat{\boldsymbol{q}}$ and the regression coefficients $\hat{\boldsymbol{\beta}}_{\mathrm{rob}}$ with smallest $\mathrm{MSE\:}\,=\,\|\hat{\boldsymbol{\Pi}}_{\bm q}^\top\boldsymbol{y}\,-\,\hat{\boldsymbol{A}}\hat{\boldsymbol{\beta}}_{\mathrm{rob}}\|_2^2$ are selected. Note that in (\ref{eq:conv_program}), the residual sum of squares is minimized instead of a robust loss, since $\hat{\bm x}=\hat{\boldsymbol{A}}\hat{\boldsymbol{\beta}}_{\mathrm{rob}}$ is already a robust estimate of the true signals and for the optimal permutation, the residuals are assumed to be normally distributed, i.e., free of outliers.

\section{Experimental Evaluation}
\label{sec:numerical_experiments}
In this section, the Shuffled Sparse Signal Recovery (SSSR) method presented in Section~\ref{ssec:rob_est} is evaluated in numerical experiments, where the goal is to reconstruct two-channel signals from noisy shuffled samples. To this end, we focus on i) reconstructing signals with DFT as in \eqref{eq:Xl}, i.e., time-domain signals that are given as low-pass filtered periodic streams of Diracs, as well as ii) an application in neuroscience, where the signal model is that of a periodic stream of decaying exponential functions. In the following, we describe the general setup of the conducted numerical experiments.
\subsection{Setup of Numerical Experiments}
\label{sec:metrics}
\paragraph*{{Experiment Pipeline}}
The signal generation, generally speaking, follows the same basic procedure. For some initialized numbers of Diracs $K_m,\: m=1,2$, associated locations $\{t_{m,k}\}_{k=0}^{K_m-1}$ and weights $\{a_{m,k}\}_{k=0}^{K_m-1}$, the two signals $\bm x_{1}$ and $\bm x_{2}$ are computed as the IDFT of their respective sequence of FS coefficients.

The locations are then generated as follows: 
First, the $K_\Sigma$ locations $\{t_{\Sigma,k}\}_{k=0}^{K_\Sigma-1}$ of the sum $\bm x_\Sigma$ are sampled from a uniform distribution on the interval $[0,1)$ with a minimum separation of $\Delta t = 0.02$ between locations to avoid arbitrarily ill-conditioned problems\footnote{Since, in theory, $\Delta t$ may be infinitesimal. For an in-depth analysis regarding bounds on $\Delta t$ that still guarantee stable recovery, see \cite{candes-fernandez-2012}.}. Subsequently, we randomly assign $K_m$ locations to $\bm x_m$, where we assume disjoint support of the signals (i.e., $K_\Sigma = K_1 + K_2$). The associated weights are then sampled from a uniform distribution on the interval $[0.5, 1]$. 

After the noiseless and unshuffled signals have been generated, zero-mean white Gaussian noise with variance $\sigma_m^2$
\begin{equation}
    \sigma_m^2\,=\,10^{\frac{\mathrm{SNR}\,-\,P_m}{10}},\:\:P_m\,=\,10\,\mathrm{log}_{10}\left(\frac{1}{N}\sum_{\ell=0}^{N-1} x_{m,\ell}^2\right),
\end{equation}
is added, where $P_m$ is the average signal power. The noisy signals $\Tilde{\boldsymbol{x}}_m$ are then shuffled resulting in the shuffled noisy signals $\Tilde{\boldsymbol{y}}_1\,=\,\bm Q\Tilde{\boldsymbol{x}}_1+(\bm I -\bm Q)\Tilde{\boldsymbol{x}}_2$, $\Tilde{\boldsymbol{y}}_2\,=\,\bm Q\Tilde{\boldsymbol{x}}_2+(\bm I - \bm Q)\Tilde{\boldsymbol{x}}_1,$ where $\bm Q=\diag(\bm q)$ is a random binary diagonal matrix.
\paragraph*{Performance Metrics}To evaluate the performance of the proposed method, we consider two different aspects: the signal reconstruction error, and the sample assignment. For the signal reconstruction, the normalized mean squared error ($\mathrm{nMSE}$) is adopted in order to compare signals with different strength/norm. To assess the performance of the sample assignment, we use the binary accuracy weighted by the absolute deviations of the two original signals $\boldsymbol{x}_{1},\:\boldsymbol{x}_{2}$. For real traces, the deviations are computed with respect to the unshuffled but noisy signals. Let $\mathcal{L}$ denote the set of indices of correctly assigned samples; the weighted accuracy ($\wa$) and the $\nmse$ are then given by
\begin{equation*}
    \label{eq:weighted_acc}
    \mathrm{WA}\,=\,\frac{\sum_{\ell'\in\mathcal{L}}|x_{1\ell'}\,-\,x_{2\ell'}|}{\sum_{\ell=0}^{N-1}|x_{1\ell}\,-\,x_{2\ell}|},\:\:\: \:\:\: \:\:\: 
    \mathrm{nMSE}=\frac{\|\boldsymbol{x}-\hat{\boldsymbol{A}}\hat{\boldsymbol{\beta}}\|_2^2}{\|\boldsymbol{x}\|_2^2}.
\end{equation*}
Since the ordering of channels is ambiguous, we have to compute the metrics with respect to both orderings and then select the better result. More concretely, the $\nmse$ has to be computed with respect to the vectors $(\boldsymbol{x}_{1}^\top,\:\boldsymbol{x}_{2}^\top)^\top$ and $(\boldsymbol{x}_{2}^\top,\:\boldsymbol{x}_{1}^\top)^\top$, respectively, while the correct assignments in the $\wa$ have to be determined with respect to $\boldsymbol{q}$ and $(\boldsymbol{1}\,-\,\boldsymbol{q})$. Then, the smaller $\nmse$ and the larger $\wa$ are selected.

\begin{figure*}
    \begin{overpic}[width=\textwidth,tics=5]{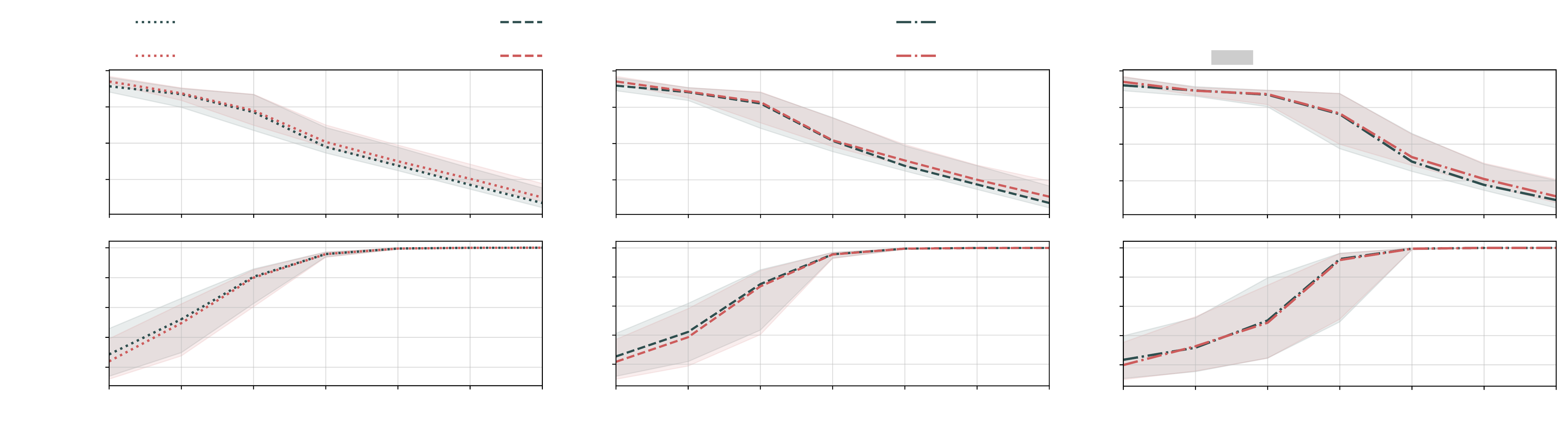}
     \put (11.75,26.25) {\small$25\%$ (known)}
    \put (11.75,24.15) {\small$25\%$ (estimated)}
    \put (35.25,26.25) {\small$33\%$ (known)}
    \put (35.25,24.15) {\small$33\%$ (estimated)}
    \put (60.25,26.25) {\small$50\%$ (known)}
    \put (60.25,24.15) {\small$50\%$ (estimated)}
    \put (80.5,24.15) {\small$25\ts{th}\text{ - }75\ts{th}$ percentile}
    \put (1,7)  {\small\rotatebox{90}{$\mathrm{WA}$}}
    \put (1,16.5)  {\small\rotatebox{90}{$\mathrm{nMSE}$}}
    \put (3, 20.75) {\small$10^{-1}$}
    \put (3, 18.4) {\small$10^{-3}$}
    \put (3, 16) {\small$10^{-5}$}

    \put (4, 11.8) {\small$1.0$}
    \put (4, 8) {\small$0.8$}
    \put (4, 4.25) {\small$0.6$}

    \put (16.75 ,0) {\small$\snr$ in $\mathrm{dB}$}
    \put (4.75, 1.75) {\small$-10$}
    \put (11.15, 1.75) {\small$0$}
    \put (15.4, 1.75) {\small$10$}
    \put (19.9, 1.75) {\small$20$}
    \put (24.5, 1.75) {\small$30$}
    \put (29.15, 1.75) {\small$40$}
    \put (33.75, 1.75) {\small$50$}

    \put (49 ,0) {\small$\snr$ in $\mathrm{dB}$}
    \put (37.1, 1.75) {\small$-10$}
    \put (43.5, 1.75) {\small$0$}
    \put (47.6, 1.75) {\small$10$}
    \put (52.2, 1.75) {\small$20$}
    \put (56.8, 1.75) {\small$30$}
    \put (61.4, 1.75) {\small$40$}
    \put (66.15, 1.75) {\small$50$}

    \put (81.4 ,0) {\small$\snr$ in $\mathrm{dB}$}
    \put (69.35, 1.75) {\small$-10$}
    \put (75.8, 1.75) {\small$0$}
    \put (80, 1.75) {\small$10$}
    \put (84.5, 1.75) {\small$20$}
    \put (89.1, 1.75) {\small$30$}
    \put (93.8, 1.75) {\small$40$}
    \put (98.4, 1.75) {\small$50$}

     \end{overpic}
    \caption{\textbf{Shuffled Spike Signals:} Performance for varying $\mathrm{SNR}$ and \textbf{known} support (gray) as well as \textbf{estimated} support (red). The left line plots show the mean weighted accuracies ($\mathrm{WA}$) and the right line plots show the median $\mathrm{nMSE}$s, where the shaded area corresponds to the $25$\ts{th} and $75$\ts{th} percentile, for different amounts of shuffled samples. The signals have been generated with the following parameters: $K_1\,=\,K_2=\,2,\:N\,=\,121.$}
    \label{fig:random_shuffle_known_est_params}
\end{figure*}
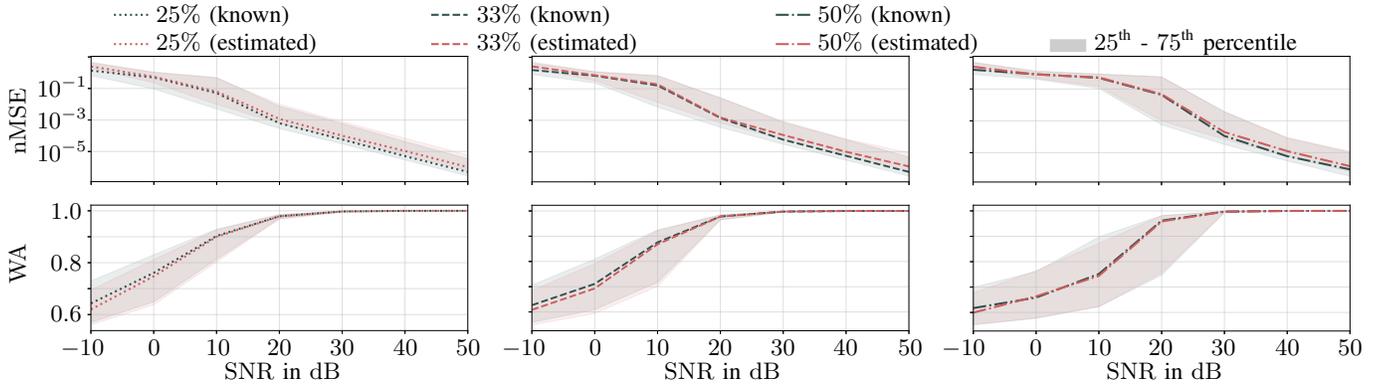

\subsection{Experiment I: Shuffled Spike Signals}
\label{sec:exp1}
First, we consider the reconstruction of two-channel signals $\bm x_1,\: \bm x_2$ that have a DFT as in \eqref{eq:Xl}, i.e., the signals are given as uniformly sampled LP filtered streams of Diracs. To investigate the influence of the propagation of errors made during the support estimation on the performance of the shuffled regression, we compare the cases where the support is known precisely and where it has to be estimated. 

We let the support of each signal consist of two shifted Diracs ($K_m = 2$), where $\{t_{m,k},\:a_{m,k}\}_{k=0}^{K_m-1}$ are sampled from uniform distributions with a minimum separation $\Delta t$ as described in Section~\ref{sec:metrics}. We consider signals with $N=121$ samples per period, and the conducted experiments consist of 1000 Monte-Carlo runs. The results are displayed in Fig.~\ref{fig:random_shuffle_known_est_params}, where the $\mathrm{SNR}$ ranges from $-10\,\si{\dB}$ to $50\,\si{\dB}$ in $10\,\si{\dB}$ steps and either $25 \%$ (dotted lines),$\:33 \%,$ (dashed lines) or $50 \%$ (dashed-dotted lines) of the samples have been shuffled. Gray and red lines correspond to the case where the support is known precisely, and where it has to be estimated, respectively.

Most noticably, we observe that the results for the case, where the support is precisely known (gray curves) and for the case, where the support has to be estimated (red curves) deviate only slightly from each other. It can be observed that for high $\snr\geq20\,\db$, the median $\nmse$ as well as the $25\ts{th}$ and $75\ts{th}$ percentiles are virtually indistinguishable. A phase transition in the $\nmse$ can be observed at $\snr\leq20\,\si{\dB}$ for any amount of shuffling. After the phase transition, the median $\nmse$ drops with approximately $10^{-1}$ per $10\,\si{\db}$ increase in $\snr$. Similarly, the median $\wa$ for $\snr\leq20\,\db$ approaches $1.00$ for any amount of shuffling, while for $10\db$ and below, a more significant performance loss can be observed.

\subsection{Experiments II: Application in Neuroimaging}
Now, we turn our attention to an application in computational neuroscience related to whole-brain calcium imaging, and evaluate the proposed method in such a setting. Calcium imaging is a technique that allows the firing activity of large populations of neurons to be measured using fluorescent calcium indicators \cite{chen_ultra-sensitive_2013}. When a neuron produces an action potential, there is an exchange of calcium ions across the cell membrane, resulting in an increase of calcium concentration $[\text{Ca}^{+2}]$ in the cell body. These fluorescent markers bind to calcium ions and therefore allow us to measure a surrogate of a neuron's firing activity. The observed fluorescent traces can be well-approximated as a stream of decaying exponentials \cite{vogelstein-etal-2010,chen_ultra-sensitive_2013}. We rely on this model and assume that the continuous-time fluorescent trace of a neuron $z(t)$ can be modelled as the convolution of a periodic stream of $K$ Dirac delta functions with a normalized period of one $s(t)$ and a causal decaying exponential shaping kernel $g(t)$ as:
\begin{equation}
\label{eq:ft}
z(t)\,=\,s(t)*g(t)\,=\,\sum_{i\in\mathbb{Z}}\sum_{k=0}^{K-1}a_k\delta(t-t_k-i)*{e^{-\alpha t}u(t)},
\end{equation}
where $a_k\in\mathbb{R}_+$ and $t_k \in [0,1)$ correspond to the amplitudes and firing instants of the neuron, $\delta(t)$ is the Dirac delta function, $u(t)$ is the unit-step function, and $\alpha>0$ is the decay factor of the exponential. Note that the FS coefficients of $z(t)$ are given by $Z_\ell = S_\ell / (\alpha + j2\pi \ell)$ for $\ell\in\mathbb{Z}$ where $S_\ell = \sum_k a_k \exp{(-j2\pi \ell t_k )}$ correspond to the FS coefficients of the stream of Diracs which follows the model in \eqref{eq:Xl}. Using the sampling setup in Fig.~\ref{fig:shuffling} and assuming that the sampling kernel $\varphi(t)$ is an ideal LP filter it can be shown \cite{bejar_finite_2020} that the sequence of samples $x_n = \left.z(t)\ast\varphi(-t)\right|_{t=n}, n=0,\ldots,N-1$ has a DFT that coincides with the FS coefficients of $z(t)$, that is $X_\ell = \sum_n z_n \exp{(-j2\pi n\ell/N)} = Z_\ell$. 
Thus, the only modification we need to make to Algorithm~\ref{algo:RTSE} so that it can continue to be used is to compensate for the attenuation of the FS coefficients caused by the shaping kernel $g(t)$.
More precisely, given shuffled observations with $\tilde{Y}_{\Sigma,\ell}=\Tilde{X}_{\Sigma,\ell}$, the estimate of the FS coefficients of the sparse signal is:
$$\hat{S}_{\Sigma,\ell} =(\alpha + j2\pi \ell)\Tilde{X}_{\Sigma,\ell},\quad \ell=0,\ldots,N-1.$$
After denoising $\hat{S}_{\Sigma,\ell}$, a line spectral estimation method can be applied to estimate the locations $t_{\Sigma,k}$ (and hence the sensing matrix) as discussed in Section~\ref{ssec:rob_est}.

\paragraph*{Signal Generation} In real calcium imaging applications, a common sampling frequency is $f_s = 30\,\si{\Hz}$, which we will assume in our numerical experiments. A thorough analysis of the decaying parameters (i.e., its corresponding intensity half-life $\tau_{1/2}$), which occur in different fluorescent proteins, has been conducted, e.g., in \cite{chen_ultra-sensitive_2013}. Their analysis shows that $\tau_{1/2}$ ranges between $0.2\, \si{\s}$ and $1\,\si{\s}$ depending on the neuron's firing rate and the used fluorescent protein. Since our signal model assumes a canonical period of $T=1$, and thus the sampling rate of the model corresponds to $N$, the decay parameter $\alpha$ is computed from $f_s$, $N$ and $\tau_{1/2}$ as $\alpha\,=\,\mathrm{ln}(2){N}/({\tau_{1/2}f_s}).$
Unless stated otherwise, each signal's locations and weights are sampled from uniform distributions on $[0,1)$ and $[0.5,1]$, respectively, as described in Section~\ref{sec:metrics}. We set $N=121$ and $\alpha\,=\,11.18$, corresponding to a half-life of $\tau_{1/2}\,=\,0.25\,\si{\s}\:$. The noiseless signals are generated as the IDFT of $Z_\ell$, and all experiments consist of $1000$ random realizations.
\subsubsection{Reconstructing Shuffled Streams of Decaying Exponentials}

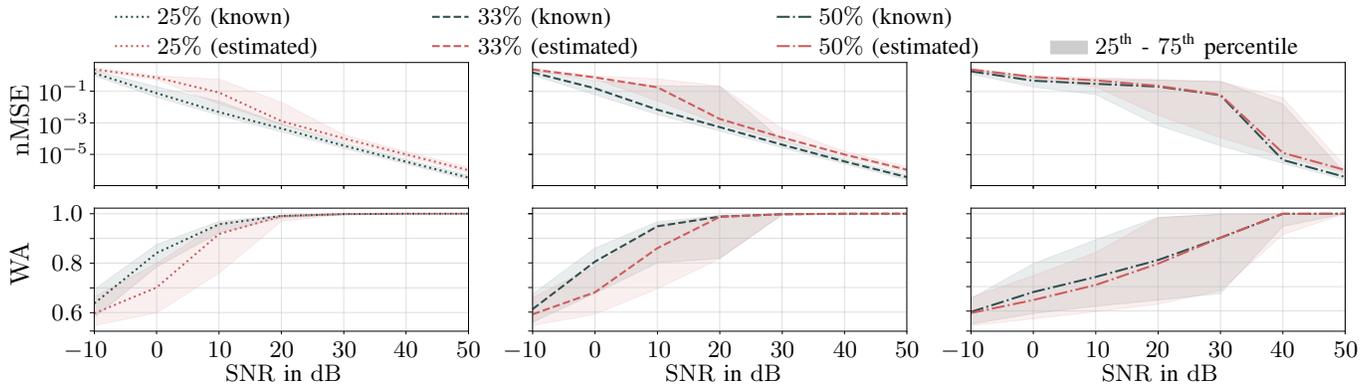
\begin{figure*}
    \begin{overpic}[width=\textwidth,tics=5]{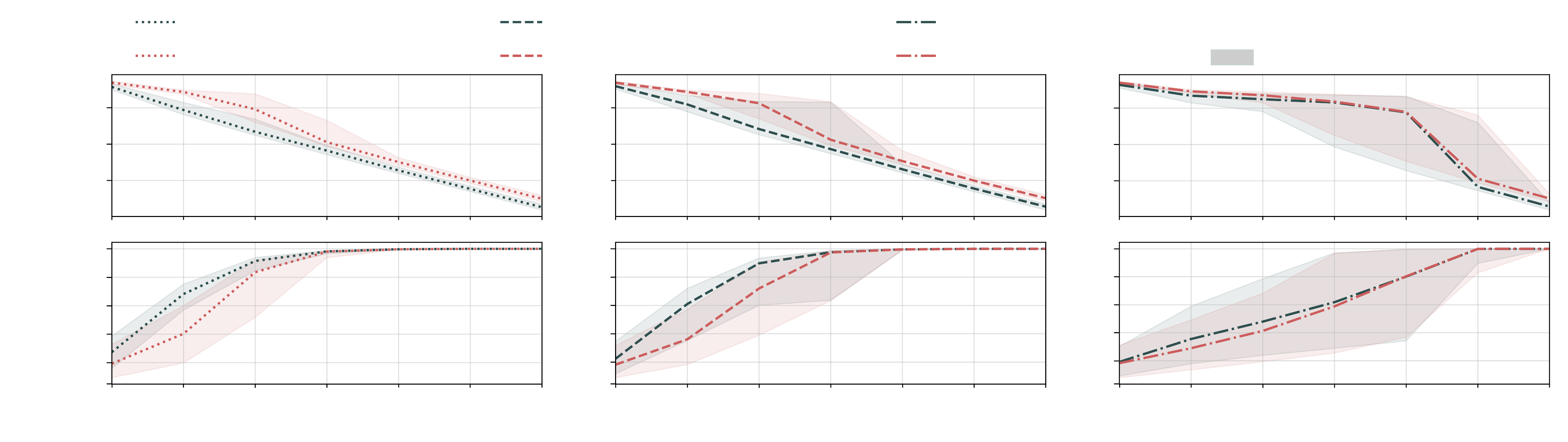}
     \put (11.75,26.25) {\small$25\%$ (known)}
    \put (11.75,24.05) {\small$25\%$ (estimated)}
    \put (35.25,26.25) {\small$33\%$ (known)}
    \put (35.25,24.05) {\small$33\%$ (estimated)}
    \put (60.25,26.25) {\small$50\%$ (known)}
    \put (60.25,24.05) {\small$50\%$ (estimated)}
    \put (80.5,24.05) {\small$25\ts{th}\text{ - }75\ts{th}$ percentile}
    \put (1,7)  {\rotatebox{90}{$\mathrm{WA}$}}
    \put (1,16.5)  {\rotatebox{90}{$\mathrm{nMSE}$}}
    \put (3, 20.6) {\small$10^{-1}$}
    \put (3, 18.2) {\small$10^{-3}$}
    \put (3, 16) {\small$10^{-5}$}

    \put (4, 11.75) {\small$1.0$}
    \put (4, 8.15) {\small$0.8$}
    \put (4, 4.5) {\small$0.6$}

    \put (16.75 ,0) {\small$\snr$ in $\mathrm{dB}$}
    \put (4.9, 1.75) {\small$-10$}
    \put (11.3, 1.75) {\small$0$}
    \put (15.45, 1.75) {\small$10$}
    \put (20, 1.75) {\small$20$}
    \put (24.6, 1.75) {\small$30$}
    \put (29.1, 1.75) {\small$40$}
    \put (33.7, 1.75) {\small$50$}

    \put (49 ,0) {\small$\snr$ in $\mathrm{dB}$}
    \put (37.1, 1.75) {\small$-10$}
    \put (43.4, 1.75) {\small$0$}
    \put (47.65, 1.75) {\small$10$}
    \put (52.1, 1.75) {\small$20$}
    \put (56.7, 1.75) {\small$30$}
    \put (61.3, 1.75) {\small$40$}
    \put (65.85, 1.75) {\small$50$}

    \put (81, 0) {\small$\snr$ in $\mathrm{dB}$}
    \put (69.2, 1.75) {\small$-10$}
    \put (75.5, 1.75) {\small$0$}
    \put (79.8, 1.75) {\small$10$}
    \put (84.2, 1.75) {\small$20$}
    \put (88.8, 1.75) {\small$30$}
    \put (93.36, 1.75) {\small$40$}
    \put (98, 1.75) {\small$50$}

     \end{overpic}
    \caption{\textbf{Shuffled Decaying Exponentials:} Performance for different levels of $\mathrm{SNR}$ and varying numbers of shuffled samples. The plots show the median as well as the $25$\ts{th} and $75$\ts{th} percentiles of the $\mathrm{nMSE}$ (top) and $\mathrm{WA}$ (bottom) for $25\%$ (left), $33\%$ (middle) and $50\%$ (right) shuffled samples for signals with the following parameters: $K_1\,= K_2\,=\,2,\: \alpha\,=\,11.18,\:N\,=\,121$.}
    \label{fig:decaying_exponentials}
\end{figure*}

In this experiment, we investigate the performance of Algorithm~\ref{algo:RTSE} for the reconstruction of simulated shuffled calcium traces (i.e., streams of shuffled decaying exponentials). For comparison, the experiments follow the same structure as with the streams of Diracs. As before, the locations of each signal are sampled from uniform distributions with a minimum separation $\Delta t$ (see Section~\ref{sec:metrics}). The results are displayed in Fig.~\ref{fig:decaying_exponentials}, where the $\mathrm{SNR}$ ranges from $-10\,\si{\dB}$ to $50\,\si{\dB}$ in $10\,\si{\dB}$ steps and either $25 \%$ (dotted lines),$\:33 \%,$ (dashed lines) or $50 \%$ (dashed-dotted lines) of the samples have been shuffled. Gray and red lines correspond to the cases where the support is known precisely and where it has to be estimated, respectively.

We observe a phase transition of the median and the percentiles of the $\nmse$ at varying amounts of shuffling, i.e., as the proportions of shuffled samples increases, the phase transition occurs for increasing $\snr$. On average, the proposed method is able to produce satisfactory results in terms of signal reconstruction and sample assignments for $\mathrm{SNR}\,\geq\,20\,\si{\dB}$ and $25\,\%$ shuffling, as well as for $\mathrm{SNR}\,\geq\,30\,\si{\dB}$ and $\approx33.3\,\%$ shuffling. For $50\,\%$ shuffled samples highly accurate results can only be achieved for very high $\snr$ ($>40\,\db$). The loss in performance ($\nmse$) from estimating the locations is significantly larger than for the spike signals, due to the attenuation of the decaying exponential kernel.
\subsubsection{Benchmark} 
\begin{figure}
    \flushright
        \begin{overpic}[width=0.47\textwidth,tics=5]{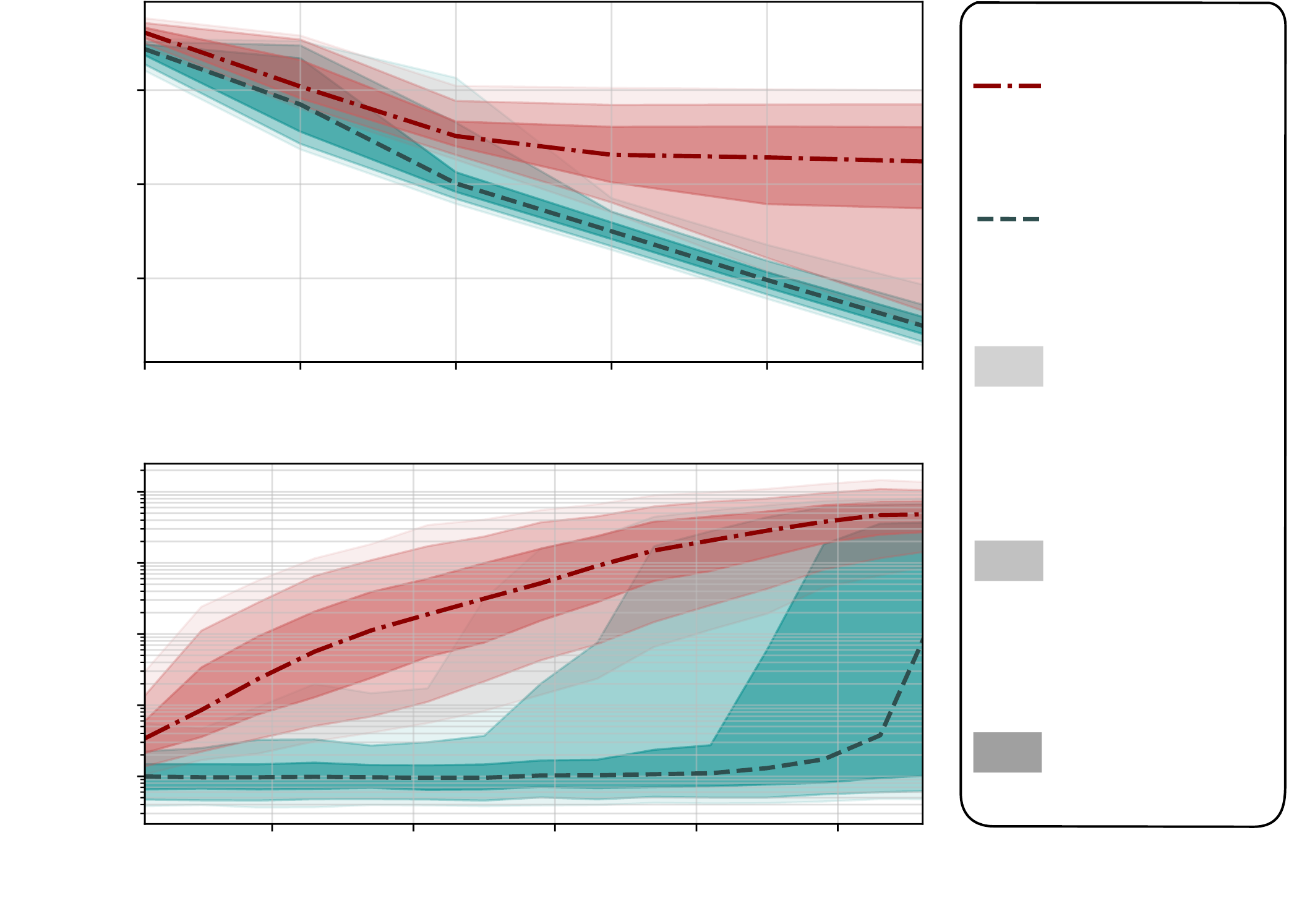}
        \put (10.5,40)  {\small$0$}
        \put (21.5,40)  {\small$10$}
        \put (33.5,40)  {\small$20$}
        \put (45.5,40)  {\small$30$}
        \put (57.75,40)  {\small$40$}
        \put (69.75,40)  {\small$50$}
        \put (32.75,36.5)  {\small$\mathrm{SNR}$ in $\mathrm{dB}$}

        \put (19.25,4)  {\small$10$}
        \put (30,4)  {\small$20$}
        \put (41,4)  {\small$30$}
        \put (52,4)  {\small$40$}
        \put (63,4)  {\small$50$}
        \put (35,0)  {\small\#Shuffled}

        \put (4,63)  {\small$10^{\text{-}1}$}
        \put (4,55.75)  {\small$10^{\text{-}3}$}
        \put (4,48.5)  {\small$10^{\text{-}5}$}
        \put (0,51)  {\rotatebox{90}{\small$\mathrm{nMSE}$}}

        \put (4,32) {\small$10^0$}
        \put (4,26.5) {\small$10^{\text{-}1}$}
        \put (4,20.95)  {\small$10^{\text{-}2}$}
        \put (4,15.5)  {\small$10^{\text{-}3}$}
        \put (4,10)  {\small$10^{\text{-}4}$}
        \put (0,16.3)  {\rotatebox{90}{\small$\mathrm{nMSE}$}}

        \put (81.5,65.75)  {\footnotesize HardEM\scriptsize\cite{abid_stochastic_2018}}
        \put (81.5,62.75)  {\footnotesize ($\mathrm{median}$)}
        \put (81.5, 55.25)  {\footnotesize Proposed}
        \put (81.5, 52)  {\footnotesize ($\mathrm{median}$)}
        \put (81.5,44)  {\footnotesize $5\ts{th}\text{ - }95\ts{th}$}
        \put (81.5,41)  {\footnotesize percentile}
        \put (81.5,29)  {\footnotesize $10\ts{th}\text{ - }90\ts{th}$}
        \put (81.5,26)  {\footnotesize percentile}
        \put (81.5,14)  {\footnotesize $25\ts{th}\text{ - }75\ts{th}$}
        \put (81.5,11)  {\footnotesize percentile}

    \end{overpic}
    \caption{\textbf{Benchmark:} Comparison of the proposed method (green) and HardEM \cite{abid_stochastic_2018} (red) for estimated locations in terms of $\mathrm{nMSE}$. In the top plot, $\approx10\%$ of the samples have been shuffled for varying $\mathrm{SNR}$. In the bottom plot, the $\mathrm{SNR}$ is fixed to $30\,\si{\dB}$, while the number of shuffled samples varies. The signals are generated according to the following parameters: $K_1\,=\,K_2\,=\,2,\:\alpha\,=\,11.18,\:N\,=\,121$.}
    \label{fig:benchmark_sig}
\end{figure}
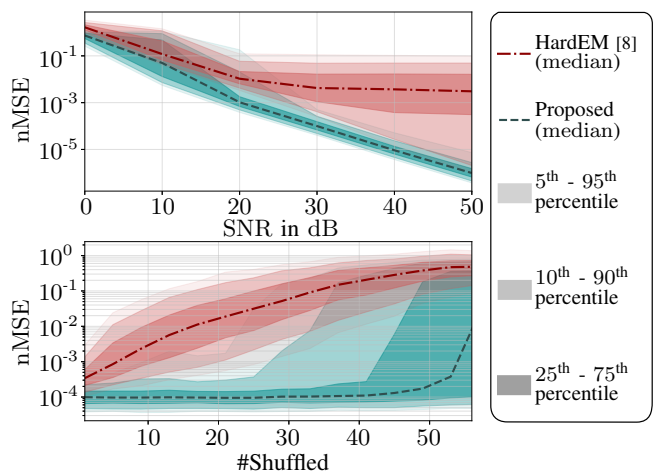
Two experiments comparing the proposed robust estimation method against the HardEM algorithm \cite{abid_stochastic_2018} for shuffled linear regression have been conducted, by replacing the shuffled regression step of the proposed algorithm (treating $\hat{\bm A}$ as ${\bm A}$) with the HardEM algorithm. The HardEM algorithm has been proposed as an approach to solve a shuffled regression problem for general permutation matrices, hence it does not take advantage of the known structure on the permutation matrices that is given in our setup.
Thus, only the $\mathrm{nMSE}$ is considered as evaluation metric since the $\mathrm{WA}$ has not been defined for general permutations. The signal generation follows the same input parameters as before. In the second experiment, the number of randomly shuffled samples increases for a fixed $\mathrm{SNR} \,=\, 30\,\si{\dB}$. The performance of both methods is displayed in Fig.~\ref{fig:benchmark_sig}, where the top plot shows the $\mathrm{nMSE}$ for a fixed number of $12$ shuffled samples and varying $\mathrm{SNR}$ ($0\,\db$ to $50\,\db$ in $10\,\db$ steps) and the bottom plot displays the $\mathrm{nMSE}$ for a fixed $\mathrm{SNR}$ of $30\,\db$ and varying amounts of shuffled samples ($1$ to $57$ in steps of $4$). While the average performance of the HardEM method already begins to deteriorate for $10$ shuffled samples, the average $\mathrm{nMSE}$ of the proposed Method does not increase significantly until more than $40$ samples are shuffled. On a similar note, for increasing $\mathrm{SNR}$ the average $\mathrm{nMSE}$ of the HardEM method saturates around $2*10^{-3}$, while for the proposed method it decreases by a constant factor of $\approx10^{-1}$ per $10\,\si{\dB}$ increase in $\mathrm{SNR}$.
\subsubsection{Refinement and Lower Bounds}

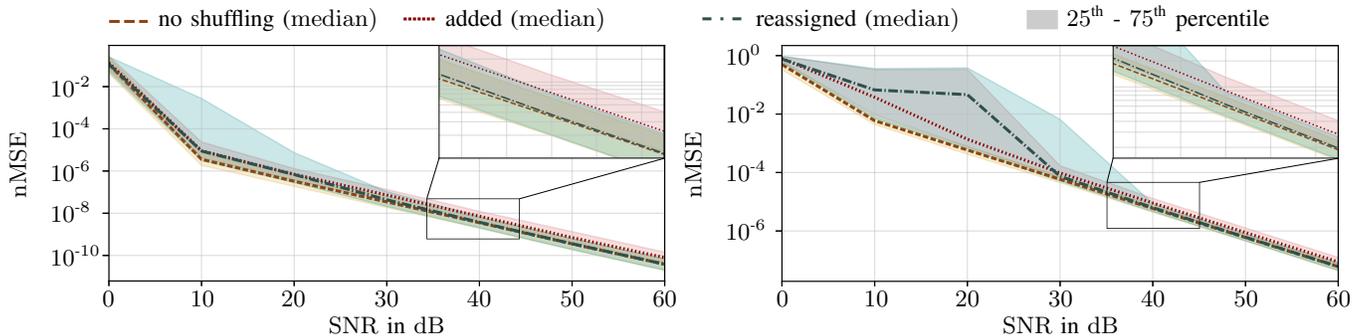
\begin{figure*}
\begin{overpic}[width=0.99\textwidth,tics=1]{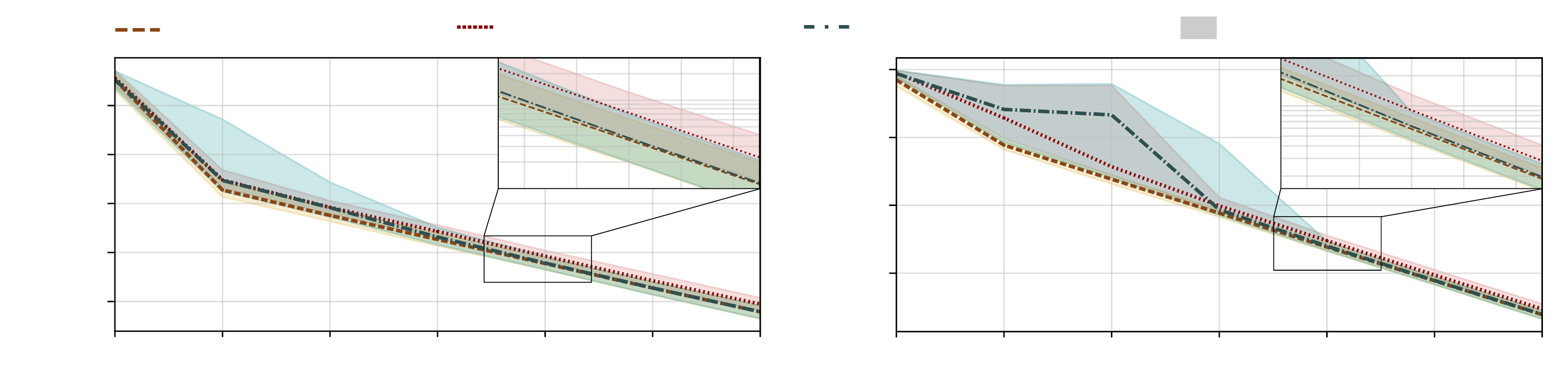}
\put (3,5)  {\small$10^{\text{-}10}$}
\put (3,8.25)  {\small$10^{\text{-}8}$}
\put (3,11.25)  {\small$10^{\text{-}6}$}
\put (3,14.5)  {\small$10^{\text{-}4}$}
\put (3,17.5)  {\small$10^{\text{-}2}$}
\put (0,10)  {\small\rotatebox{90}{$\mathrm{nMSE}$}}

\put (53,6.9)  {\small$10^{\text{-}6}$}
\put (53,11.25)  {\small$10^{\text{-}4}$}
\put (53,15.6)  {\small$10^{\text{-}2}$}
\put (53,19.9)  {\small$10^{0}$}
\put (50,10)  {\small\rotatebox{90}{$\mathrm{nMSE}$}}

\put (56.75,2)  {\small$0$}
\put (63.25,2)  {\small$10$}
\put (70,2)  {\small$20$}
\put (76.9,2)  {\small$30$}
\put (83.75,2)  {\small$40$}
\put (90.65,2)  {\small$50$}
\put (97.5,2)  {\small$60$}
\put (73.5,0)  {\small$\mathrm{SNR}$ in $\mathrm{dB}$}

\put (6.85,2)  {\small$0$}
\put (13.35,2)  {\small$10$}
\put (20.1,2)  {\small$20$}
\put (27,2)  {\small$30$}
\put (33.85,2)  {\small$40$}
\put (40.75,2)  {\small$50$}
\put (47.6,2)  {\small$60$}
\put (23.6,0)  {\small$\mathrm{SNR}$ in $\mathrm{dB}$}

\put (78.25,22.75)  {\small$25\ts{th}\text{ - }75\ts{th}$ percentile}
\put (55.25, 22.75)  {\small reassigned ($\mathrm{median})$}
\put (32.25, 22.75)  {\small added ($\mathrm{median})$}
\put (11, 22.75)  {\small no shuffling ($\mathrm{median})$}

\end{overpic}
\caption{\textbf{Refinement and Lower Bounds:} Left and right plots correspond to approximately $33\%$ shuffled samples. (Left) Normalized MSE of the estimated support using unshuffled individual sample vectors (brown), added sample vectors (red) and reassigned individual sample vectors according to the estimated assignment (blue). (Right) Signal reconstruction errors when: sample vectors are not shuffled (brown); Algorithm~\ref{algo:RTSE} is applied to shuffled sample vectors, i.e., the support is estimated from added samples (red); Step 2 (unlabeled sensing) of Algorithm~\ref{algo:RTSE} is applied to reassigned samples and the support is estimated from individual channels after the reassignment (blue). The signals were generated according to the following parameters: $K_1\,=\,K_2\,=\,2,\:\alpha\,=\,11.18,\:N\,=\,121$.}
\label{fig:reest}
\end{figure*}
A natural approach to possibly increase the SSSR methods's performance is to introduce an additional refinement step, where the support is estimated from individual channels with reassigned samples according to the estimate $\hat{\bm \Pi}_{\bm q}$ resulting from Algorithm~\ref{algo:RTSE}. This new estimate of the matrix $\bm A$ and the reassigned samples then provide the input to Step 2 (shuffled regression) of the SSSR. This section addresses this by comparing such a refinement with the estimate from Algorithm~\ref{algo:RTSE} and with the estimate using spectral line estimation on $\tilde{\bm x}_1,\;\tilde{\bm x}_2$, which provides a lower bound on the performance of the proposed methods. The evaluation accounts for the performance w.r.t. support estimation as well as the signal reconstruction.

Each signal consists of two decaying exponentials, where $40\;(\approx33\%)$ samples are shuffled, and the $\mathrm{SNR}$ ranges from $0\,\si{\dB}$ to $60\,\si{\dB}$. The results are displayed in Fig.~\ref{fig:reest}, where the left plot shows the $\mathrm{nMSE}$ of the estimated locations resulting from applying the ADMM denoising and Prony's method to individual unshuffled samples $\tilde{\bm x}_1$, $\tilde{\bm x}_2$ (brown dashed), added samples $\tilde{\bm y}_\Sigma$ (red dotted) and reassigned individual samples $\hat{{\bm x}}_1 = \hat{\bm Q}\tilde{\bm y}_1+(\bm I - \hat{\bm Q})\tilde{\bm y}_2$, $\hat{{\bm x}}_2 = \hat{\bm Q}\tilde{\bm y}_2+(\bm I - \hat{\bm Q})\tilde{\bm y}_1$ (blue dashed-dotted). For any $\mathrm{SNR}\,\geq\,10\,\si{\dB}$, the estimation error for individual channels is approximately half of the $\mathrm{nMSE}$ for estimates using $\tilde{\bm x}_\Sigma$, i.e., the refinement may lead at best to a $3\,\si{\dB}$ gain. Indeed, we observe that the median $\mathrm{nMSE}$ of the estimated locations from reassigned samples approaches the lower bound as the $\mathrm{SNR}$ increases, as can be seen in the zoom-in in top right corner of the plot. 

The right plot shows the performance of the entire signal reconstruction, when samples are correctly assigned and the support is estimated from $\tilde{\bm x}_1$, $\tilde{\bm x}_2$ (brown dashed), Algorithm~\ref{algo:RTSE} is applied to shuffled samples $\tilde{\bm y}$ (red dotted) and when these results have been refined by estimating the support from reassigned samples $\hat{{\bm x}}_1$, $\hat{{\bm x}}_2$, which serve together with the estimated support as input to Step 2 (unlabeled sensing) of Algorithm~\ref{algo:RTSE} (blue dashed-dotted). Similar to the support estimation, the refinement leads at most to a $3\,\si{\dB}$ gain. However, due to the propagation of errors between Step 1 and Step 2 of the SSSR, we observe larger variation in the performance of Algorithm~\ref{algo:RTSE} and its refinement. A phase transition in the median $\mathrm{nMSE}$ occurs around $20\,\si{\dB}$, and the $\mathrm{nMSE}$ shows high variation even at $30\,\si{\dB}$. However, for more than $30\,\db\:\snr$, the lower bound is being tightly approached by the assignment as can be observed in the zoom-in in the top-right corner of the plot.

Although this experiment shows that for moderate levels of noise and shuffling as well as larger amounts of shuffled samples and high $\mathrm{SNR}$, a refinement does lead to an increase in performance of signal reconstruction, the gain to be made is at most $3\,\si{\dB}$. 
Note that for $M > 2$ signals, the loss in performance may be larger due to the oversampling factor (in the unshuffled case) increasing as a multiple of $M$ and for equal noise variances $\sigma_\Sigma^2 = \sum_m\sigma_m^2 = M\sigma_1^2$. However, such an analysis is beyond the scope of this work.

\subsubsection{Artificially Shuffled Real Data}

\begin{figure*}
    \flushright
    \begin{overpic}[width=0.9\textwidth,tics=5]{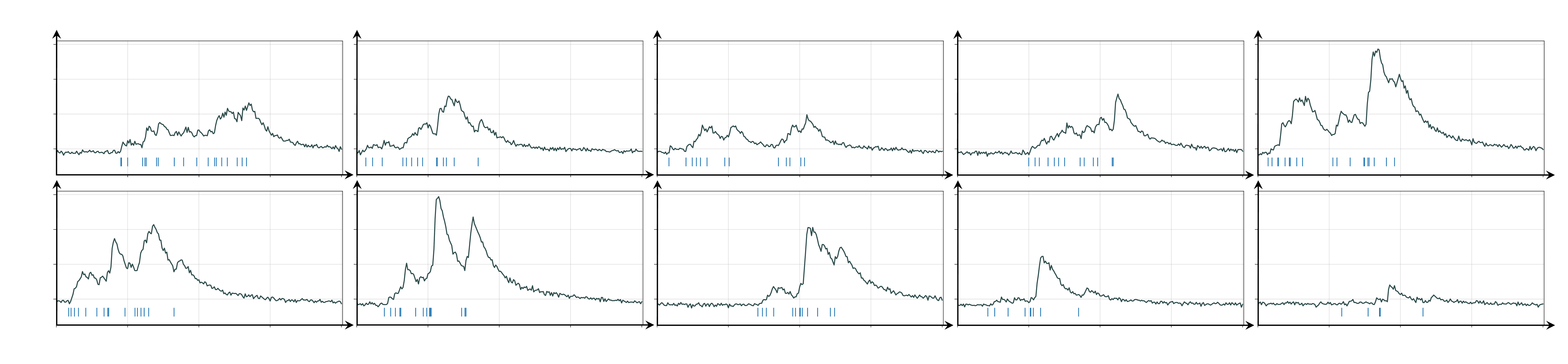}
        \put (1,21.5) {\small(a)}
        \put (1,4) {\footnotesize \rotatebox{90}{$\Delta F/F_0$}}
        \put (1,13.5) {\footnotesize \rotatebox{90}{$\Delta F/F_0$}}
        \put (12.5,1) {\footnotesize t}
        \put (31.65,1) {\footnotesize t}
        \put (50.8,1) {\footnotesize t}
        \put (69.9,1) {\footnotesize t}
        \put (89.05,1) {\footnotesize t}
    \end{overpic}
    \centering
    \begin{overpic}[width=0.88\textwidth,tics=5]{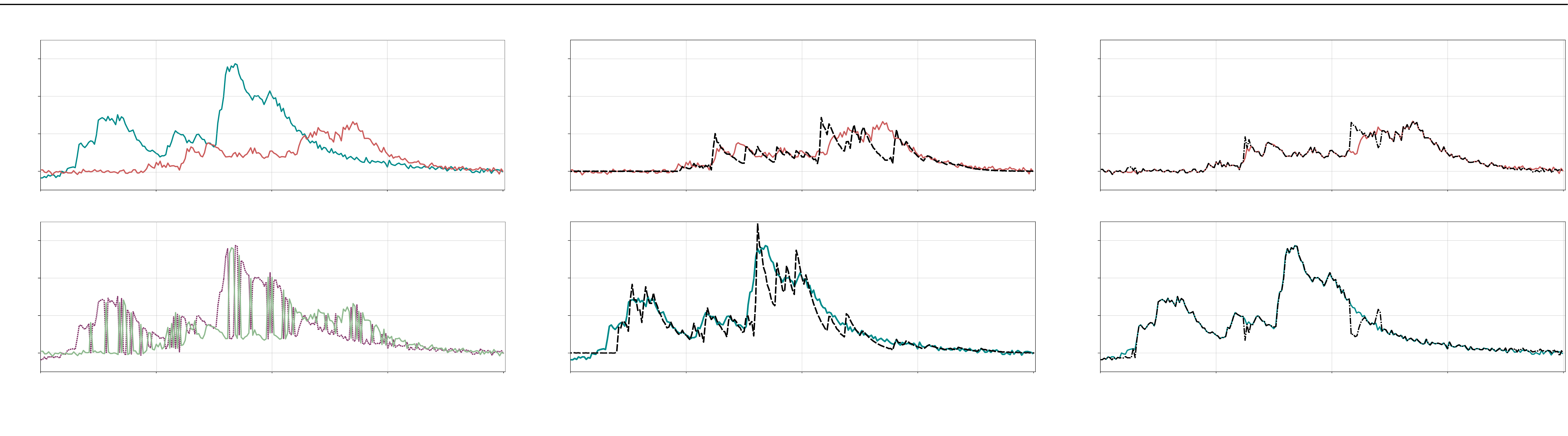}
        \put (0, 25.5) {\small(b)}
        \put (0,6) {\footnotesize \rotatebox{90}{$\Delta F/F_0$}}
        \put (0, 17.5) {\footnotesize \rotatebox{90}{$\Delta F/F_0$}}
        \put (17.05, 2.5) {\footnotesize t}
        \put (50.95, 2.5) {\footnotesize t}
        \put (84.65, 2.5) {\footnotesize t}
        \put (12, 25.8) {\footnotesize True Traces}
        \put (12, 14.1) {\footnotesize Shuffled (25$\%$)}
        \put (45.5, 25.8) {\footnotesize Reconstruction}
        \put (46, 0.7) {\footnotesize $R^2=0.8088$}
        \put (80.15, 25.8) {\footnotesize Assignment}
        \put (78.7, 0.7) {\footnotesize $\mathrm{WA}=0.9444$}
        \put (5, 0.7) {\footnotesize $\alpha_{est}=7.029$ (\footnotesize $\tau_{est}=0.396\,s$)}
    \end{overpic}
    \flushleft
    \caption{\textbf{Artificially Shuffled Real Data:} (a) Collection of approx. $4\,\si{\s}$ long windows of true traces (relative intensity of the fluorescent traces $\Delta F/F_0$). Blue vertical bars show the locations of action potentials extracted from electrophysiological recordings; (b) Example of two randomly selected unshuffled traces given by green and red solid lines (top left). Specified amount of samples are shuffled at random (bottom left); Example of reconstructed signals with associated $R^2$ given by dashed lines (middle); Example of estimated sample assignment with associated $\mathrm{WA}$ given by dashed-dotted lines (right).}
    \label{fig:real_data}
\end{figure*}
\label{ssec:real_data}
\begin{figure}
    \flushright
    \begin{overpic}[width=0.48\textwidth,tics=5]{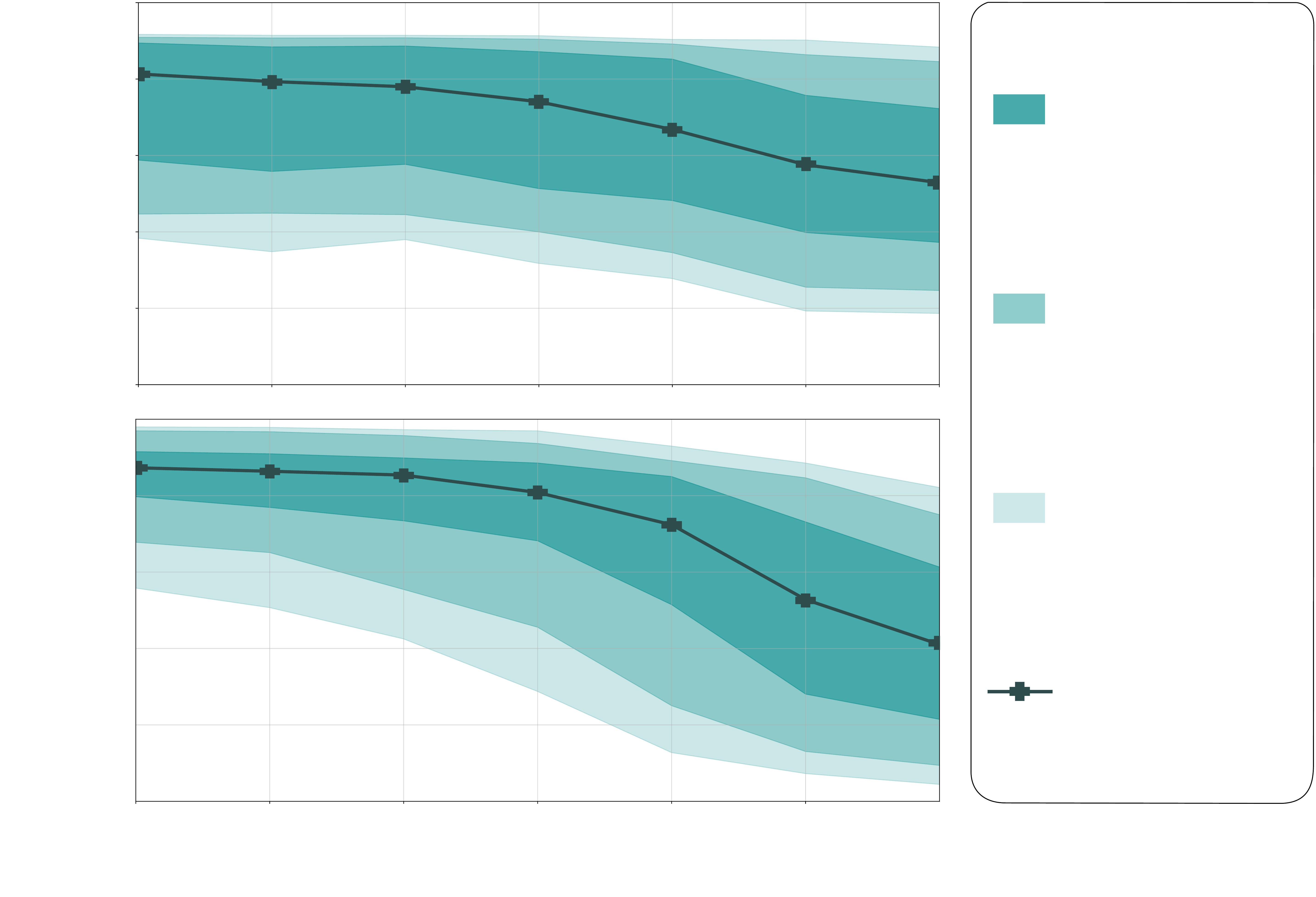}
        \put (80.5, 60.2) {$25\ts{th}$ - $75\ts{th}$}
        \put (80.5, 57) {percentile}
        \put (80.5, 45.2) {$10\ts{th}$ - $90\ts{th}$}
        \put (80.5, 42) {percentile}
        \put (80.5, 30.2) {$5\ts{th}$ - $95\ts{th}$}
        \put (80.5, 27) {percentile}
        \put (80.75, 14.8) {median}
        
        \put (0, 52) {\rotatebox{90}{$R^2$}}
        \put (4, 29.5) {{0.9}}
        \put (4, 23.5) {{0.8}}
        \put (4, 17.5) {{0.7}}
        \put (4, 11.75) {{0.6}}

        \put (0, 18) {\rotatebox{90}{$\mathrm{WA}$}}
        \put (4, 61) {{0.8}}
        \put (4, 55.1) {{0.6}}
        \put (4, 49.5) {{0.4}}
        \put (4, 43.5) {{0.2}}

        \put (33, 0) {\#Shuffled}
        \put (18.5, 3.5) {$20$}
        \put (28.65, 3.5) {$40$}
        \put (38.75, 3.5) {$60$}
        \put (49.1, 3.5) {$80$}
        \put (58.2, 3.5) {$100$}
        \end{overpic}
        
     \caption{\textbf{Artificially Shuffled Real Data - Results:} Average results for $1000$ Monte Carlo runs for different amounts of shuffled samples. Top plot shows the median and percentiles of $R^2$, while the bottom plot shows $\mathrm{WA}$.}
    \label{fig:real_data_results}
\end{figure}
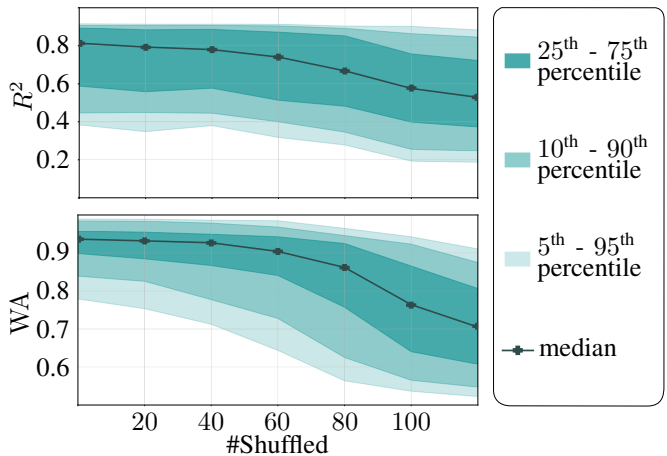
In this section, we apply the proposed method on artificially shuffled calcium imaging traces from mice taken from the cai-1 dataset \cite{genie_project_simultaneous_2015}. Before beginning the estimation, we initially subtract the baselines of the individual signals. The dataset includes simultaneous electrophysiological recordings, from which the spike counts are extracted. Then, $\alpha$ is estimated from the sum of the signals by means of the golden section search algorithm as described in \cite{bejar_finite_2020}. With this estimate of the decay and the extracted spike counts $K_\Sigma$, the two-step SSSR method can be used to estimate the signals and unshuffle the samples. Note that for this application, the weights of the spikes can only take nonnegative values. However, the MM-estimator as an unconstrained estimator may yield some negative coefficients due to the presence of noise. Simply forcing these coefficients to zero results, however, in a bad fit of the estimated signals ($R^2\approx0.2$). Instead, we enforce nonnegativity by replacing the least squares solution in each iteration of the IRWLS algorithm with a nonnegative LS solution. Fig.~\ref{fig:real_data}~(a) shows the traces that are used in the experiment. Each trace consists of $241$ samples; at a sampling frequency of $60\,\si{\Hz}$ this yields a length of approximately $4\,\si{\s}$. In each trial, two traces are randomly selected from this set and shuffled at random according to the specified number of shuffled samples which is illustrated in the left plots of Fig.~\ref{fig:real_data}~(b). Then, the SSSR is applied to the resulting shuffled traces, returning an assignment of the samples and an estimate of the true signals, as shown in the middle and right plots of Fig.~\ref{fig:real_data}~(b). Since the true signals are unknown, we report the $R^2$ as a measure of the signal fit and the weights of the $\mathrm{WA}$ are computed w.r.t. the traces. The average results for $1000$ trials are shown in Fig.~\ref{fig:real_data_results}, where the amount of shuffled samples is varied. On average, we observe a highly accurate performance of the proposed method up to $80$ shuffled samples ($\approx33\,\%$), which matches the observations made in the numerical simulations for moderate levels of noise. Interestingly, the average weighted accuracy of the estimated assignment is above $0.9$ even for larger disagreements between the estimated model and the true traces, suggesting a higher tolerance of the assignment against even larger deviations of the model from ground truth.
\section{Conclusion}
\label{sec:conclusion}
We introduced a framework for the reconstruction of shuffled multi-channel signals termed structured unlabeled sensing, and we derived sufficient conditions for the existence of a unique solution up to an ordering of the channels. We showed that this framework also applies to shuffled multi-channel signals that admit a sparse representation in an overcomplete dictionary, where the sensing matrix is not precisely known a priori. For such signals, we proposed a two-step procedure for shuffled sparse signal recovery (SSSR) that combines sparse signal estimation and robust regression and illustrated its effectiveness in an application related to calcium imaging. Our methodology could be generalized to sparse signal representations other than the ones considered in this work and find application in a variety of real-world problems, such as multiple target tracking or contactless vital signs monitoring of multiple people, where imprecise measurement and channel assignments are present. However, further work is required to establish performance guarantees and to consider scenarios involving a higher number of channels.

\small

\bibliographystyle{ieeetr}
\bibliography{bibliography.bib}

\end{document}